\documentclass{article}
\usepackage[affil-it]{authblk}
\usepackage{times}
\usepackage{helvet}
\usepackage{courier}
\usepackage{amssymb,amsmath,amsthm}
\usepackage{comment}
\usepackage{graphicx}

\usepackage{amsmath}
\usepackage{comment}
\usepackage{amssymb}

\newcommand{\npprob}[4]{
\begin{center}
\begin{tabular}{|r p{9cm} |}
\hline
\multicolumn{2}{|l|}{\textsc{#1}}\\
\textit{Instance:}& #2.\\
\textit{Parameter:}& #3.\\
\textit{Problem:}& #4\\
\hline
\end{tabular}
\end{center}
}

\newcommand{\confversion}[1]{}

\newtheorem{theorem}{Theorem}[section]

\newtheorem{corollary}[theorem]{Corollary}
\newtheorem{example}[theorem]{Example}

\newtheorem{proposition}[theorem]{Proposition}
\newtheorem{prop}[theorem]{Proposition}
\newtheorem{lemma}[theorem]{Lemma}

\newtheorem{remark}[theorem]{Remark}

\newtheorem{newremarkcore}[theorem]{Remark}

\newtheorem{definitioncore}[theorem]{Definition}

\newenvironment{definition}
  {\begin{definitioncore}\rm}
  {\end{definitioncore}}

\newcommand{\rela}{\mathbf{A}}
\newcommand{\relb}{\mathbf{B}}
\newcommand{\relc}{\mathbf{C}}
\newcommand{\reld}{\mathbf{D}}
\newcommand{\relt}{\mathbf{T}}
\newcommand{\relp}{\mathbf{P}}
\newcommand{\relg}{\mathbf{G}}
\newcommand{\relm}{\mathbf{M}}

\newcommand{\nats}{\mathbb{N}}

\newcommand{\fancya}{\mathcal{A}}
\newcommand{\fancyg}{\mathcal{G}}
\newcommand{\fancym}{\mathcal{M}}
\newcommand{\fancyr}{\mathcal{R}}

\newcommand{\T}{\mathcal{T}}
\renewcommand{\P}{\mathcal{P}}
\newcommand{\C}{\mathcal{C}}
\newcommand{\BIN}{\mathcal{B}}

\newcommand{\dP}{\overrightarrow{\P}}
\newcommand{\dC}{\overrightarrow{\C}}
\newcommand{\dBIN}{\overrightarrow{\BIN}}

\newcommand{\dom}{\mathsf{dom}}

\newcommand{\tw}{\mathsf{tw}}
\newcommand{\td}{\mathsf{td}}
\newcommand{\pw}{\mathsf{pw}}

\newcommand{\qr}{\mathsf{qr}}
\newcommand{\ar}{\mathsf{ar}}

\newcommand{\justHOM}[1]{\textsc{Hom}(#1)}
\newcommand{\HOMP}[1]{p\textsc{-Hom}(#1)}
\newcommand{\EMP}[1]{p\textsc{-Emb}(#1)}

\newcommand{\SHOMP}[1]{p\textsc{-\#Hom}(#1)}

\newcommand{\EVAL}{\textsc{EVAL}}

\newcommand{\core}{\mathsf{core}}

\newcommand{\pl}{\textup{pl}}

\newcommand{\A}{\mathbb A}

\newcommand{\B}{\mathbb B}

\newcommand{\PATH}{\textup{PATH}}
\newcommand{\TREE}{\textup{TREE}}

\newcommand{\ignore}[1]{}

\begin{document}

\title{The Fine Classification of Conjunctive Queries 
and\\Parameterized Logarithmic Space Complexity}

\author{Hubie Chen\thanks{E-mail address: \texttt{hubie.chen@ehu.es}}}
\affil{Departamento LSI, Facultad de Inform\'{a}tica\\
Universidad del Pa\'{i}s Vasco, 
San Sebasti\'{a}n, Spain\\
and\\
IKERBASQUE, Basque Foundation for Science,
Bilbao,
Spain
}

\author{Moritz M\"uller\thanks{E-mail address: \texttt{moritz.mueller@univie.ac.at}}}
\affil{
Kurt G\"{o}del Research Center, Universit\"{a}t Wien, Vienna, Austria
}
\date{}
\maketitle

\begin{abstract}
We perform a fundamental investigation of the complexity of conjunctive query evaluation from the perspective of parameterized complexity.  
We classify sets of boolean conjunctive queries according to the complexity of
this problem.  Previous work showed that a set of conjunctive queries
is fixed-parameter tractable precisely when the set is equivalent to a
set of queries having bounded treewidth.  We present a fine
classification of query sets up to parameterized logarithmic space
reduction.  We show that, in the bounded treewidth regime, there are
three complexity degrees and that the properties that determine the
degree of a query set are bounded pathwidth and bounded tree depth.
We also engage in a study of the two higher degrees via logarithmic
space machine characterizations and complete problems.  Our work
yields a significantly richer perspective on the complexity of
conjunctive queries and, at the same time, 
suggests new avenues of research in parameterized complexity.
\end{abstract}

\section{Introduction} \label{section:introduction}

Conjunctive queries are the most basic and most heavily studied database
queries, and can be formalized logically as formulas consisting of a
sequence of existentially quantified variables, followed by a
conjunction of atomic formulas.  
Ever since the landmark 1977 article of 
Chandra and Merlin~\cite{ChandraMerlin77-optimal},
complexity-theoretic aspects of conjunctive queries have been a
research subject of persistent and enduring interest which
continues to the present day (as a sampling, we point to the works~\cite{AbiteboulHullVianu95-foundationsdatabases,koalititsvardicontainment,PapadimitriouYannakakis99-database,GottlobLeoneScarcello01-complexityacyclic,GottlobLeoneScarcello02-hypertreedecomposisions,Grohe07-otherside,CreignouKolaitisVollmer08-overview,SchweikardtSchwentickSegoufin09-querylanguages,Marx10-tractablehypergraph};
see the discussions and references therein for more information).
The problem of evaluating a conjunctive query on a relational database
is equivalent to a number of well-known problems, including 
conjunctive query containment, the homomorphism problem
on relational structures, and the 
constraint satisfaction problem~\cite{ChandraMerlin77-optimal,koalititsvardicontainment}.
That this evaluation problem appears in many equivalent guises
attests to the fundamental and primal nature of this problem, and
it has correspondingly been approached and studied
from a wide variety of perspectives and motivations.
The resulting literature has not only been fruitful in terms of
continually
providing insights into and notions for understanding conjunctive
queries themselves, but has also meaningfully 
%contributed to and 
fed back into
a 
richer understanding of computational complexity theory at large,
and of common complexity classes in particular.
This is witnessed by the observation that various flavors of
conjunctive query evaluation are used as prototypical complete
problems
for complexity classes such as NP and W[1] 
(refer, for example, to the books by 
Creignou, Khanna, and Sudan~\cite{CreignouKhannaSudan01-boolean} and 
by Flum and Grohe~\cite{flumgrohe}, respectively).  
Another example of this phenomenon is the work showing 
LOGCFL-completeness of evaluating acyclic conjunctive queries
(as well as of many related problems) due to 
Gottlob, Leone, and Scarcello~\cite{GottlobLeoneScarcello01-complexityacyclic}.

As has been eloquently articulated in the literature~\cite{PapadimitriouYannakakis99-database}, 
the employment of classical complexity notions such as 
polynomial-time tractability
to grade the complexity of conjunctive query evaluation is not totally
satisfactory.
For in the context of databases, the typical scenario
is the evaluation of a relatively short query on a relatively
large database;
this suggests a notion of time complexity
wherein a non-polynomial dependence on the query 
may be tolerated, so long as the dependence on the database
is polynomial.
Computational complexity theory has developed and studied
precisely such a relaxation of polynomial-time tractability,
called \emph{fixed-parameter tractability}, in which
arbitrary dependence in a \emph{parameter} is permitted;
in our query evaluation setting, the query size is normally
taken as the parameter.
The class of such tractable problems is denoted by FPT.
Fixed-parameter tractability is the base tractability notion of
\emph{parameterized complexity theory}, a comprehensive theory for
studying problems where each instance has an associated parameter.
As a parameterized problem, conjunctive query evaluation is complete for
the parameterized complexity class 
W[1]~\cite{PapadimitriouYannakakis99-database,flumgrohe};
the property of W[1]-hardness plays, in the parameterized setting,
a role similar to that played by NP-hardness in the classical setting.

%;this indicates intractability of the problem even in the parameterized
%setting, as
%the class W[1] is believed
%to be strictly larger than FPT, the class of fixed-parameter tractable
%problems.

Due to the general intractability of conjunctive query evaluation,
a recurring theme in the study of conjunctive queries is the
identification of structural properties that provide
tractability; such properties include \emph{acyclicity} and 
\emph{bounded treewidth}~\cite{GottlobLeoneScarcello01-complexityacyclic,koalititsvardicontainment}.
A natural research issue is to obtain a systematic understanding
of what properties ensure tractability, by classifying 
all sets of queries
according to the complexity of the evaluation problem.
We focus on boolean conjunctive queries, which, in logical parlance,
are queries without free variables.
Formally, let $\Phi$ be a set of boolean conjunctive queries,
and define $\EVAL(\Phi)$ to be the problem of deciding,
given a query $\phi \in \Phi$ and a relational structure $\relb$,
whether or not $\phi$ evaluates to true on $\relb$.
One can then inquire for which sets $\Phi$ the problem $\EVAL(\Phi)$
is tractable.
For mathematical convenience, %reasons of presentation, 
we use an equivalent
formulation
of this problem.  It is known that 
each boolean conjunctive query $\phi$ can be bijectively represented
as a relational structure $\rela$ in such a way that, for any relational
structure $\relb$, it holds that $\phi$ is true on $\relb$ if and only
if
there exists a homomorphism from $\rela$ to $\relb$~\cite{ChandraMerlin77-optimal}.
Hence, the following family of problems is equivalent to 
the family of problems $\EVAL(\Phi)$.  Let $\fancya$ be a set of
structures,
and denote by $\justHOM{\fancya}$ the problem of deciding,
given a structure $\rela \in \fancya$ and a second structure $\relb$,
whether or not there is a homomorphism from $\rela$ to $\relb$.
Use $\HOMP{\fancya}$ to denote the parameterized version of this
problem,
where the size of $\rela$ is taken as the parameter.

Under the assumption that the structures in $\fancya$ have bounded
arity,
Grohe~\cite{Grohe07-otherside} 
presented a classification of the tractable problems
of this form: if the \emph{cores} of $\fancya$ have bounded treewidth, then
%$\justHOM{\fancya}$ is polynomial-time tractable,
%which implies that 
the problem $\HOMP{\fancya}$ is fixed-parameter
tractable;
otherwise, the problem $\HOMP{\fancya}$ is W[1]-hard.
The \emph{core} of a structure can be intuitively thought of as a smallest
equivalent structure.
Grohe's classification thus shows that, in the studied setting, 
the condition of bounded
treewidth
is the \emph{only} property guaranteeing tractability
(assuming FPT $\neq$ W[1]).
Recall that treewidth is a graph measure which, intuitively speaking,
measures the similitude of a graph to a tree, with a lower measure
indicating
a higher degree of similarity.
The assumption of bounded arity provides robustness
in that translating between two reasonable representations of 
structures can be done efficiently; this is in contrast to the
case of unbounded arity, where the choice of representation
can dramatically affect complexity~\cite{ChenGrohe10-succinct}.

The present article was motivated by the following fundamental 
research question:
\emph{What algorithmic/complexity behaviors of conjunctive queries are
possible, within the regime of fixed-parameter tractability?}
That is, we endeavored to obtain a finer perspective on the 
parameterized complexity
of conjunctive queries, and in particular, 
on the possible sources of tractability thereof,
 by presenting a classification result akin to
Grohe's,
but for queries that are fixed-parameter tractable.
%We embarked 
As is usual in computational complexity, 
we make use of a weak notion of reduction in order to be able to make
fine distinctions within the tractable zone.  
Logarithmic space computation is 
a common machine-based mode of computation that is often used
to make distinctions within polynomial time;
correspondingly, we adopt \emph{parameterized logarithmic space computation},
which is obtained by relaxing logarithmic space computation
much in the way that fixed-parameter tractability is obtained by
relaxing
polynomial time,
as the base complexity class 
and as the reduction notion used in our investigation.

We present a classification theorem %(Theorem~\ref{?}) 
that 
comprehensively describes, 
for each set $\fancya$ of structures having bounded arity and boun\-ded
treewidth,
the complexity of the problem
$\HOMP{\fancya}$, up to parameterized logarithmic space reducibility
(Section~\ref{sect:classification}).
Let $\T$ denote the set of all graphs that are trees, 
$\P$ denote the set of all graphs that are paths,
and,
for a set of structures $\fancya$,
 let $\fancya^*$ denote the set of structures obtainable by taking
a structure $\rela \in \fancya$ and adding each element of $\rela$ as
a relation.
Our theorem shows that precisely three degrees of behavior are possible:
such a problem $\HOMP{\fancya}$ is either equivalent to
$\HOMP{\T^*}$, equivalent to $\HOMP{\P^*}$, or is solvable in
parameterized logarithmic space (Theorem~\ref{thm:classification}).
Essentially speaking, bounded pathwidth and bounded tree depth
are the properties that determine which of the three cases hold;
as with treewidth, both pathwidth and tree depth are graph measures
that associate a natural number with each graph.
A key component of our classification theorem's proof is a reduction
that, in effect, allows us to prove hardness results
on a problem $\HOMP{\fancya}$ based on the hardness of
$\HOMP{\mathcal{M}^*}$ where $\mathcal{M}$ consists of certain
graph minors derived from $\fancya$
(Lemma~\ref{lem:redrow}).
The proof of our classification theorem utilizes this reduction
in conjunction with excluded minor characterizations of
graphs of bounded pathwidth and of bounded tree depth.
We remark that, in combination with the \emph{excluded grid theorem}
from graph minor theory, the discussed reduction can be employed
to readily derive Grohe's classification from the hardness 
of the \emph{colored grid homomorphism problem}; this hardness
result was presented by Grohe, Schwentick, and Segoufin~\cite{GroheSchwentickSegoufin01-conjunctivequeries}.
A fascinating aspect of our classification theorem,
which is shared with that of Grohe,
is that natural graph-theoretic conditions--in our case,
those of bounded pathwidth
and bounded tree depth--arise naturally as the relevant properties
that are needed to present our classification.
%%%as opposed to being externally
%%%imposed to obtain positive results (as is often the case).
This theorem also widens the interface among conjunctive queries,
graph minor theory, and parameterized complexity that is
present in the discussed work~\cite{GroheSchwentickSegoufin01-conjunctivequeries,Grohe07-otherside}.

Given that the problems 
$\HOMP{\P^*}$ and
$\HOMP{\T^*}$
are the \emph{only} problems (up to equivalence) 
above parameterized logarithmic space
that emerge from our classification, we then seek a richer
understanding
of these problems.
In particular, we engage in a study of the
complexity classes that these problems define:
we study the class of problems that reduce to $\HOMP{\P^*}$,
and likewise for $\HOMP{\T^*}$
(Sections~\ref{sect:path} and~\ref{sect:tree}).
Following a time-honored tradition in complexity theory,
we present machine-based definitions of these classes, which classes
we
call $\PATH$ and $\TREE$, respectively.
The machine definition of $\PATH$ comes from recent work of 
Elberfeld, Stockhusen, and Tantau~\cite{tantau} 
and is based on nondeterministic Turing machines
satisfying two simultaneous restrictions: first, that only
parameterized logarithmic space is consumed;
second, that the number of nondeterministic bits used is bounded,
namely,
by 
the product of 
the logarithm of the input size
and a constant depending on the parameter.
The machine characterization of $\TREE$ is similar, but it is based
on alternating Turing machines where, in addition to the
nondeterministic bits permitted previously, 
a parameter-dependent number of conondeterministic 
bits may also be used.
In addition to proving that the problems 
$\HOMP{\P^*}$
and
$\HOMP{\T^*}$
are complete for the machine-defined classes,
we also prove that for any set of
structures
$\fancya$ having bounded pathwidth, the 
parameterized \emph{embedding} problem
$\EMP{\fancya}$ is in $\PATH$, and prove an analogous result
for structures of bounded treewidth and the class $\TREE$.

In the final section of the paper, we present a fine classification for the
problem
of counting homomorphisms which is analogous to our classification
for the homomorphism problem (Section~\ref{sect:counting}).

Our work shows that the complexity classes $\PATH$ and $\TREE$
are heavily populated with complete problems, and,
along with the recent % work of Elberfeld, Stockhusen, and Tantau
work~\cite{tantau},
 suggests the 
further development of the study of space-bounded parameterized
complexity~\cite{describing,alternation} 
and, speaking more broadly, the study of complexity classes
within FPT, which may include classes based on circuit or parallel models of
computation.  
We can mention the following natural structural
questions.
Are either of the classes $\PATH$ or $\TREE$ closed under complement?
Can any evidence be given either in favor of or against such closure?
Even if the classes $\PATH$ and $\TREE$ are not closed under complement, could it be that
co-$\PATH \subseteq \TREE$?  
Another avenue for future research is to develop the theory of the degrees of counting problems identified by our counting classification.
We shall mention some further open questions in the final section.

%in that reasonable representations of the structures will tend to result

%(The $p$ in $\HOMP{\fancya}$ indicates the parameterized version
%of this problem.)

% EVAL
% HOM
% for convenience of presentation...
% bounded arity...

%A comprehensive classification 

%Grohe theorem.

\begin{comment}

Our starting point is the following result of Grohe.

\begin{theorem}[Grohe]
\end{theorem}

need to acknowledge: Dalmau and Jonsson, who similarly proved a version of Grohe's theorem for the counting problems.

Moved in from Classification section:

From the classification theorem $\HOMP{\T^*}$ and $\HOMP{P^*}$ emerge as natural degrees with respect to pl-reducibility. Are 
these degrees different? Proving this is presumably difficult: observe that  $\HOMP{\P^*}\not\equiv_\pl\HOMP{\T^*}$ implies 
$\HOMP{\T^*}\notin\textup{paraL}$ and hence $\textup{FPT}\neq\textup{paraL}$ and hence $\textup{P}\neq\textup{L}$.

To understand the question we study 
the corresponding complexity classes, namely the problems pl-reducible to $\HOMP{\P^*}$ on the one hand 
and to $\HOMP{\T^*}$ on the other hand. These two classes are studied in the following two subsections respectively.
As usual in complexity theory we try to understand the classes by 
giving machine characterizations of them as well as by exhibiting complete problems. 
\end{comment}

%This article is the full version of 
%an extended abstract which will appear in the proceedings of
%the 2013 Symposium on Principles of Database Systems (PODS).

\section{Preliminaries}
\label{sect:preliminaries}

For $n\in\mathbb N$ we define $[n]:=\{1,\ldots,n\}$ if $n>0$ and $[0]:=\emptyset$.
We write $\{0,1\}^{\le n}$ for the set of binary strings $x\in\{0,1\}^*$ of length $|x|\le n$;
we have $\{0,1\}^{\le 0}=\{\lambda\}$ where $\lambda$ is the empty string.

\subsection{Structures, homomorphisms and cores}
\subsubsection*{Structures} A {\em vocabulary } $\tau$ is a finite set of  relation symbols, where each $R\in\tau$
has an associated {\em arity} $\ar(R)\in\mathbb N$. 
 A {\em $\tau$-structure} $\rela$  consists of 
a nonempty finite set $A$, its \emph{universe},  together with an 
{\em interpretation} $R^{\mathbf A}\subseteq A^{\ar(R)}$ of every
$R\in\tau$. 
Let us emphasize that, in this article, we consider only finite structures. 
A {\em substructure (weak substructure)} of $\rela$ is a structure {\em induced} by a nonempty subset $X$ of $A$, i.e. the structure 
$\langle X\rangle^\rela$ with universe $X$ that interprets every $R\in\tau$ by (respectively, a subset of)
$X^{\ar(R)}\cap R^{\rela}$.
A {\em restriction} of a structure is obtained by forgetting the interpretations of some symbols, and an {\em expansion} 
of a structure is obtained by adding interpretations of some symbols. 
%The {\em direct product} of two $\tau$-structures $\rela$ and $\relb$ is denoted $\rela\times\relb$; it 
%has universe $A\times B$ and interprets $R\in\tau$ by $\{((a_1,b_1),\ldots,(a_{\ar(R)},b_{\ar(R)}))\mid \bar a\in R^\rela,\bar b\in R^\relb\}$.
We view {\em directed graphs} as $\{E\}$-structures $\mathbf G:=(G,E^{\mathbf G})$ 
for binary $E$;  $\relg$ is a {\em graph} if $E^{\mathbf G}$ is irreflexive and symmetric. Note that a weak substructure of a graph is a subgraph. The graph 
{\em underlying} a directed graph $\mathbf G$ without loops (i.e.\ with irreflexive $E^\relg$) is 
obtained by replacing $E^{\mathbf G}$ with its symmetric closure.
We shall be concerned with the following classes of structures. 
\begin{enumerate}\itemsep=0pt
\item[--] 
For $k \geq 2$, 
the structure $\overrightarrow{\mathbf P_k}$ has universe $[k]$ and edge relation $\{ (i, i+1)\mid i \in [k-1] \}$. 
The class $\dP$ of {\em directed paths} consists of the structures that are isomorphic to a structure of this form.

Let $\mathbf P_k$ be the graph underlying~$\overrightarrow{\mathbf P_k}$. 
The class $\P$ of {\em paths} consists of the structures that are isomorphic to a structure of this form.

\item[--] For $k \geq 2$,
the structure $\overrightarrow{\mathbf C_k}$ has universe
$[k]$ and edge relation $\{ (i, i+1)\mid i \in[k-1]\}\cup\{(k,1)\}$. The class 
$\dC$ of {\em directed cycles} 
consists of the structures that are isomorphic to a structure of this form.

Let $\mathbf C_k$ be the graph underlying $\overrightarrow{\mathbf C_k}$. The class
$\C$ of {\em cycles} consists of the structures that are isomorphic to a structure of this form.

\item[--] 
For $k\ge 0,$ the structure $\overrightarrow{\mathbf B_k}$ has
universe  $\{0, 1 \}^{\leq k}$ and binary 
relations $S_i^{\overrightarrow{\mathbf B_k}} = \{ (x, xi)\mid x \in \{ 0, 1 \}^{\le k-1} \}$ for $i\in\{0,1\}$.
The class $\dBIN$  consists of the structures that are isomorphic to a structure of this form.

Let $\relt_k$ be the graph underlying the directed graph 
$(\{ 0, 1 \}^{\leq k},S_0^{\overrightarrow{\mathbf B_k}}\cup S_1^{\overrightarrow{\mathbf B_k}})$.
%The class $\BIN$ of {\em complete binary trees} 
%consists of the structures that are isomorphic to a structure of this form.

Let $\relb_k$ be the structure with universe
$\{ 0, 1 \}^{\leq k}$ and binary relations
$S_0^{\relb_k}$, $S_1^{\relb_k}$ 
defined to be the symmetric closures
of
the relations
$S_0^{\overrightarrow{\relb_k}}$, $S_1^{\overrightarrow{\relb_k}}$,
respectively.
The class $\BIN$ consists of the structures that are isomorphic to 
a structure of this form.

\item[--] Finally,  $\T$ is the class of {\em trees}, 
that is, the class of connected, acyclic graphs.

%\item[--] $\fancyg^*$ where $\fancyg$ is the set of grids, i.e. graphs %$\relg=(G,E^G)$ such that
%$G=[k]\times[k]$ for some $k\ge 1$ and $E^G=\{((i,j),(i',j'))\mid %|i-i'|+|j-j'|=1\}$.
\end{enumerate}
A class of structures $\fancya$ has {\em bounded arity} if there exists a $r\in\mathbb N$ such that any relation symbol interpreted 
in any structure $\rela\in\fancya$ has arity at most $r$.

\subsubsection*{Homomorphisms}
Let $\rela$, $\relb$ be structures.
A {\em homomorphism} from $\rela$ to $\relb$ is a 
function $h:A\to B$ such that  for all $R\in\tau$ and for all $\bar a=(a_1, \ldots, a_{\ar(R)}) \in R^{\rela}$ it holds that $h(\bar a)\in R^{\relb}$ where we write $h(\bar a)=(h(a_1), \ldots, h(a_{\ar(R)}))$. 
A {\em partial} homomorphism from $\rela$ to $\relb$ is the empty set or a homomorphism from a substructure of $\rela$ to $\relb$; 
equivalently, this is a partial function $h$ from $A$ to $B$ that is a homomorphism from $\langle\dom(h)\rangle^{\rela}$ to $\relb$ 
if the domain $\dom(h)$ of $h$ is not empty.
As has become usual in our context, by an {\em embedding} we mean an injective homomorphism.

A structure  $\rela$ is a {\em core} if all homomorphisms from $\rela$ to $\rela$ are embeddings. 
Every structure $\rela$ maps homomorphically to a weak substructure of itself which is a core. This weak substructure 
is unique up to isomorphism and called the core {\em of} $\rela$ (cf.~\cite{federvardi}).
For a set of structures $\mathcal A$ we let $\core(\mathcal A)$ denote the set of cores of structures in $\mathcal A$.
It is not hard to see that two structures $\rela,\relb$ are homomorphically equivalent (that is, there are homomorphisms in both directions) if and only if they have the same core.

When $\rela$ is  a structure, we use $\rela^*$ to denote its
expansion that interprets for every $a\in A$ a fresh unary relation symbol $C_a$
by $C_a^{\rela^*}=\{a\}$. 
For a class of structures $\fancya$ we let
$$
\fancya^*:=\{\rela^*\mid\rela\in\fancya\}.
$$

\begin{example}\label{exa:core} \em  
The following facts are straightforward to verify.
Trees with at least two vertices and cycles of even length have a single edge as core, and so do cycles of even length. 
Cycles of odd length are cores, and so are directed paths. Structures of the form $\rela^*$ are cores.
\end{example}

\subsection{Notions of width}
We rely on Bodlaender's survey \cite{bodlaendersurvey} as a general
reference for the notions of treewidth and pathwidth. 
Tree depth was introduced in~\cite{nesetrilmendez}. 
%Recently there has been some renewed interest in this notion in the context of model-checking monadic second order logic \cite{grohetd,shrub}, and, as we shall see, it is also central to understand the space complexity of homomorphism problems.

A {\em  tree-decomposition} of a graph $\relg=(G,E^\relg)$ is a pair 
%$(\relt,(X_t)_{t\in T})$ 
of a tree $\relt$ and 
a family of {\em bags} $X_t\subseteq G$ for $t\in T$ such that $G=\bigcup_{t\in T}X_t, E^\relg\subseteq\bigcup_{t\in T} X^2_t$ 
and $X_t\cap X_{t'}\subseteq X_{t''}$ whenever $t''$ lies on the simple path from $t$ to~$t'$; it is called a {\em path-decomposition}
if $\relt$ is a path; its {\em width} is $\max_{t\in T}|X_t|-1$.  

The {\em treewidth} $\tw(\relg)$ of $\relg$
 is the minimum width of a tree-decomposition of $\relg$. The {\em pathwidth} 
$\pw(\relg)$ of $\relg$ is the minimum width of a path-decomposition of~$\relg$.

By a {\em rooted tree} $\relt$ we mean an expansion $(T,E^{\relt},\textit{root}^\relt)$ of a tree $(T,E^{\relt})$ by a unary 
relation symbol $\textit{root}$ interpreted by a singleton containing the {\em root}. 
The {\em tree depth} $\td(\relg)$ of $\relg$ is the minimum $h\in\mathbb N$ such that every 
connected component of $\relg$ is a subgraph of the closure of some rooted tree of height $h$. Here, 
the {\em closure} of a rooted tree is obtained by adding an edge from $t$ to $t'$ whenever $t$ lies on the 
simple path from the root to $t'$.

The tree depth $\td(\rela)$ of an arbitrary structure $\rela$ is the tree depth of its {\em Gaifman graph}:
it has vertices $A$ and an edge between $a$ and $a'$ if and only if $a$ and $a'$ are 
different and occur together in some tuple in some relation in $\rela$. The notions 
$\pw(\rela)$ and $\tw(\rela)$ are similarly defined. 

A class $\fancya$  of structures has {\em bounded tree depth} if there is $w\in\mathbb N$ such that
$\td(\rela)\le w$ for all $\rela\in\fancya$. Having bounded pathwidth or treewidth is similarly explained.
It is not hard to see that bounded pathwidth is implied by bounded tree depth, and, trivially, bounded treewidth 
is implied by bounded pathwidth. 
The converse statements fail:

\begin{example}\label{exa:unbounded} \em The class $\P$ has unbounded tree depth and bounded pathwidth (cf.~\cite[Lemma~2.2]{nesetrilmendez}). 
The class $\mathcal B$ has unbounded pathwidth and bounded treewidth (see. e.g.~\cite[Theorem~67]{bodlaendersurvey}).
\end{example}

Such classes are characterized as those excluding certain minors as follows. 
The first two statements are well-known from Robertson and Seymour's graph minor series (cf.~\cite[Theorems~12,13]{bodlaendersurvey}) 
and the third is from~\cite[Theorem~4.8]{blumensath}.

\begin{theorem}\label{thm:rs} Let $\mathcal C$ be a class of graphs.
\begin{enumerate}\itemsep=0pt
\item \textup{(Excluded Grid Theorem)} $\mathcal C$ has bounded treewidth if and only if $\mathcal C$ excludes some grid as a minor.
\item \textup{(Excluded Tree Theorem)} $\mathcal C$ has bounded pathwidth if and only if $\mathcal C$ excludes some tree 
as a minor.
\item \textup{(Excluded Path Theorem)} $\mathcal C$ has bounded tree depth if and only if $\mathcal C$ excludes some path as a minor.
\end{enumerate}
\end{theorem}

A class of graphs $\mathcal C$ {\em excludes} a graph $\relm$ {\em as a minor} if $\relm$ is not a minor of any graph in $\mathcal C$.
Recall, $\relm$ is a {\em minor} of a graph $\relg$ 
%if it can be obtained from a subgraph of  $\relg$ by contracting edges.
%Equivalently, $\mathbf M$ is a minor of $\relg$ if and only 
if there exists a {\em minor map} $\mu$ from $\mathbf M$ to $\relg$, that is, a family
$(\mu(m))_{m\in M}$ of pairwise disjoint, non-empty, connected subsets of $G$ such that for all
$(m, m') \in E^{\relm}$ there are $v \in \mu(m)$ and $v' \in \mu(m')$ with
$(v, v') \in E^{\relg}$.
%
% from $M$ into the power set of $G$ such that
%\begin{enumerate}\itemsep=0pt
%\item[(a)] for all $m \in M$, the set $\mu(m)$ is non-empty and connected in $\relg$;
%
%\item[(b)] for all  $m, m' \in M$ with $m \neq m'$, the sets $\mu(m)$ and $\mu(m')$ are disjoint; and,
%
%\item[(c)] for all edges $(m, m') \in E^{\relm}$ there exist $v \in \mu(m)$ and $v' \in \mu(m')$ such that
%$(v, v') \in E^{\relg}$.
%\end{enumerate}

It is easy to verify that $\td, \pw,\tw$ are monotone with respect to the minor pre-order, 
that is, e.g. $\td(\relg)\ge\td(\mathbf M)$ for every minor $\relm$ of
$\relg$. 
Example~\ref{exa:unbounded}
thus gives the (easy) directions from left to right in the above theorem.

%\pagebreak

\subsection{Parameterized complexity}

\subsubsection*{Turing machines} We identify (classical) problems with sets $Q\subseteq\{0,1\}^*$ of finite binary strings.
We use Turing machines with a (read-only) input tape and several worktapes as our basic model of computation.
We will consider nondeterministic and alternating Turing machines 
with binary nondeterminism and co-nondeterminism. 
For concreteness, let us agree that a nondeterministic machine has a special {\em (existential) guess state}; a configuration 
 with the guess state has two successor configurations obtained by changing the guess state to one out of two further
 distinguished states $s_0,s_1$.  An alternating machine may additionally have a {\em universal guess state} that follows a similar convention. 
For a function  $f:\{0,1\}^*\to\mathbb N$ we say that $\mathbb A$ {\em uses $f$ (co-)nondeterministic bits} 
if for every input $x\in\{0,1\}^*$ every run of $\mathbb A$ on $x$ contains at most $f(x)$ many configurations with the existential
 (respectively, universal) guess state.

\subsubsection*{Fixed-parameter (in)tractability} A {\em parameterized problem} $(Q,\kappa)$ is a pair of a classical 
problem $Q\subseteq\{0,1\}^*$ and a  logarithmic space computable
{\em parameterization} $\kappa:\{0,1\}^*\to\mathbb N$ associating with any instance $x\in\{0,1\}^*$ 
its {\em parameter} $\kappa(x)\in\mathbb N$.\footnote{Usually polynomial time is allowed to compute $\kappa$
but as we are interested in parameterized logarithmic space we adopt a more restrictive notion as~\cite{tantau}. 
Natural parameterizations are often simply projections.} 
%
%For example, the parameterized clique problem is the classical clique problem together with the parameterization $\kappa$ that maps a binary string encoding a pair $(\relg,k)$ to $k$, and all other strings to, say, 0. We present parameterized problems in the following form:
%
%\npprob{$p$-Clique}{a graph $G$ and a natural $k\in\mathbb N$}{$k$}{does $G$ contain a $k$-clique?}
%
%
%Let $\kappa$ be a parameterization. 
%
A Turing machine  is {\em fpt-time bounded (with respect to $\kappa$)} if on input $x\in\{0,1\}^*$ it 
runs in time  $f(\kappa(x))\cdot|x|^{O(1)}$ where $f:\mathbb N\to\mathbb N$ is a computable function.
The class FPT (para-NP) contains the parameterized problems $(Q,\kappa)$ such that $Q$ is decided 
(accepted) by an fpt-time bounded deterministic (nondeterministic) Turing machine.
An {\em fpt-reduction} from $(Q,\kappa)$ to $(Q',\kappa')$ is a reduction $R:\{0,1\}^*\to\{0,1\}^*$ from $Q$ to $Q'$ 
that is computable by a fpt-time bounded (with respect to $\kappa$) Turing machine and such 
that $\kappa'\circ R\le f\circ \kappa$ for some computable $f$.

%Here, we just recall the following well-known fact.
%
%\begin{theorem}\label{theo:clique} $p${\sc -Clique} is $\textup{W[1]}$-complete under fpt-reductions.
%\end{theorem}

We are concerned with homomorphism and embedding problems associated with classes of structures $\fancya$.
\npprob{$\HOMP{\fancya}$}{A pair of structures $(\rela, \relb)$
where  $\rela\in\fancya$}{$|\rela|$}{Is there a homomorphism from $\rela$ into $\relb$?}
\npprob{$\EMP{\fancya}$}{A pair of structures $(\rela, \relb)$
where  $\rela\in\fancya$}{$|\rela|$}{Is there an embedding from $\rela$ into $\relb$?}
These problem definitions exemplify
 how we present parameterized problems. More formally, the parameterization indicated is the function that maps a string encoding a pair of structures $(\rela,\relb)$ to $|\rela|$, and any other string to, say, 0. Here, $|\rela|:=|\tau|+|A|+\sum_{R\in\tau}|R^A|\cdot\ar(R)$ is the {\em size} of $\rela$; note that the length of a 
reasonable binary encoding of~$\rela$ is $O(|\rela|\cdot\log|A|)$ (cf.~\cite{ffg}).

%It can be shown that $\HOMP{\fancya}$ and $\EMP{\fancya}$ are contained in W[1]. For $\fancya$ the class of finite cliques %Theorem~\ref{theo:clique} readily implies that $\HOMP{\mathcal A},\EMP{\mathcal A}$ and are $\textup{W[1]}$-hard and hence complete under fpt-reductions.

The theory of parameterized intractability is centered around the W-hierarchy,
which consists of the classes
$
\textup{W[1]}\subseteq \textup{W[2]}\subseteq\cdots\subseteq\textup{W[P]}.
$
The class W[P] contains the parameterized problems $(Q,\kappa)$ that are accepted by nondeterministic
Turing machines that are fpt-time bounded with respect to $\kappa$ and use
$f(\kappa(x))\cdot\log |x|$ many nondeterministic bits. 
We refer to the monographs \cite{flumgrohe,downeyfellows} for more
information about the W-hierarchy.  
It is well-known that,
when $\fancya$ is a decidable class of structures,
the problems
$\HOMP{\fancya}$ and $\EMP{\fancya}$ are contained in W[1];
when $\fancya$ is the e.g. class of cliques, 
these problems are W[1]-hard and hence W[1]-complete under fpt-reductions.

\subsubsection*{Parameterized logarithmic space} A Turing machine is {\em parameterized logarithmic space bounded (with respect to $\kappa$)}, in short,
{\em pl-space bounded (with respect to~$\kappa$)} if on input $x\in\{0,1\}^*$ it runs in space 
$O(f(\kappa(x))+\log n)$,
where $f:\mathbb N\to\mathbb N$ is some computable function. The class $\textup{para-L}$ (para-NL)
contains the parameterized problems $(Q,\kappa)$ such that $Q$ is decided (accepted) 
by a (non)deterministic Turing machine  that is pl-space bounded with respect to $\kappa$. Obviously, 
$$
\textup{para-L}\subseteq\textup{para-NL}\subseteq\textup{FPT}\subseteq
\textup{W[P]}\subseteq\textup{para-NP}.
$$

 %
%Note this  gives a pl-space bounded algorithm for formulas of constant width. More interestingly, Flum and Grohe showed that one gets  pl-space bounded algorithms when restricting to bounded degree graphs~\cite[Example~6]{describing}.\marginpar{H: i propose to move this to  right before it is used} Their proof shows
%
%\begin{theorem}[\cite{describing}] The following parameterized problem is in $\textup{paraL}$.
%
%\npprob{}{a graph $\mathbf G$ of degree at most $k$ and a first-order sentence $\varphi$}{$k+|\varphi|$}{$\mathbf G\models\varphi$ ?}
%\end{theorem}

\begin{remark}\em  
%\marginpar{H: move this to appendix?}
%\marginpar{I would propose to delete (outcomment) it in case we need the space; it would be nice if we could keep it}
Allowing in the above definition space $f(\kappa(x))\cdot\log |x|$ gives strictly 
larger classes known as (the stronlgy uniform versions of) XL and XNL. These classes are likely to be incomparable with FPT: they do not contain FPT 
unless $\textup{P}=\textup{NL}$ and contain problems that are even AW[SAT]-hard under fpt-reductions. We shall not be concerned with these 
classes here and refer the interested reader to \cite{alternation,tantau} for proofs of the mentioned facts and further information.  \cite{describing} 
gives some general account of the para- and X-operators.
\end{remark}

%%For later reference we state an alternative definition of paraNL, mimicking the characterization of NL by
%%logspace verifiers from the classical world.
%% 
%%A {\em verifier } is a deterministic Turing machine with a second input tape, called a {\em guess tape}. The %crucial restriction is that the head %on the guess tape only moves right. It is run on pairs $(x,y)\in\{0,1\}^*$ 
%%of strings where initially $x$ is written on the input tape and $y$ is  written on the guess tape. It is a verifier %{\em for} a problem %$Q\subseteq\{0,1\}^*$  if for all $x\in\{0,1\}^*$
%%\[
%%x\in Q\Longleftrightarrow\exists y\in\{0,1\}^{*}: \mathbb V\text{ accepts } (x,y).
%%\]
%%Space bounds for verifiers apply only on the worktapes and especially not on the guess tape.
%%It is well-known and easy to see that precisely the problems in NL have verfiers that are logarithmic space %bounded~\cite{arorabarak}. It is %routine to verify the following.
%%
%%\begin{proposition}\label{prop:verifier} A  parameterized problem $(Q,\kappa)$ is in $\textup{paraNL}$ if %and only if there exists a verifier %%for $Q$ that is parameterized logarithmic space bounded with respect to %$\kappa$.
%%\end{proposition}

Let $\kappa$ be a parameterization. A function $F:\{0,1\}^*\to\{0,1\}^*$ is {\em implicitly pl-computable 
(with respect to $\kappa$)} if the parameterized problem
\npprob{Bitgraph$(F)$}{A triple $(x, i, b)$ where $x\in \{0,1\}^*$, 
$i\ge 1$, and $b\in\{0,1\}$}{$\kappa(x)$}{Does $F(x)$ have length $|F(x)| \geq i$ 
and $i$th bit equal to $b$?}
is in para-L. The following is straightforwardly verified as in the
classical setting
of logarithmic space computability.

\begin{lemma}\label{lem:trans} Let $\kappa,\kappa'$ be parameterizations and let $F,F':\{0,1\}^*\to\{0,1\}^*$ be implicitly 
 pl-computable with respect to $\kappa$ and $\kappa'$ respectively. Then $F'\circ F$ is implicitly pl-computable with respect to $\kappa$.
\end{lemma}

Let $(Q,\kappa),(Q',\kappa')$ be parameterized problems. 
A {\em pl-reduction} from $(Q,\kappa)$ to $(Q',\kappa')$ is 
a reduction $R:\{0,1\}^*\to\{0,1\}^*$ from $Q$ to $Q'$ that is  implicitly pl-computable\footnote{
It is routine to verify that $F$ is implicitly pl-computable if and only if it is computable by a pl-space
 bounded Turing machine with a write-only output tape. 
Our definition is equivalent to the ones in~\cite{describing,alternation,tantau}.
}  
with respect to $\kappa$ and such that there exists a  computable function $f:\mathbb N\to\mathbb N$
such that $\kappa'\circ R\le f\circ \kappa$.
We write $(Q,\kappa)\le_\pl (Q',\kappa')$ to indicate that such a reduction exists. We write $(Q,\kappa)\equiv_\pl (Q',\kappa')$ if both
$(Q,\kappa)\le_\pl (Q',\kappa')$ and $(Q',\kappa')\le_\pl (Q,\kappa)$. 
%It follows immediately from Lemma~\ref{lem:trans} that $\le_\pl$ is transitive and $\equiv_\pl$ is an equivalence relation.

%We say $(Q,\kappa)$ is {\em pl-Turing reducible} to $(Q',\kappa')$ and write $(Q,\kappa)\le^T_\pl (Q',\kappa')$ if there are 
%a pl-space bounded (with respect to $\kappa$) Turing machine $\mathbb A$ 
%with oracle $Q'$ that decides $Q$, and a computable $f$ such that on every input $x\in\{0,1\}^*$ all 
%queries $y\stackrel{?}{\in} Q'$ 
%of $\A$ on $x$ have parameter $\kappa'(y)\le f(\kappa(x))$. We write $(Q,\kappa)\equiv^T_\pl (Q',\kappa')$ if both
%$(Q,\kappa)\le^T_\pl (Q',\kappa')$ and $(Q',\kappa')\le^T_\pl (Q,\kappa)$.

\section{Classification}
\label{sect:classification}

\begin{theorem}[Classification Theorem]\label{thm:classification}
Let $\fancya$ be a decidable class of structures of bounded arity
such that $\core(\fancya)$ has bounded tree\-width.
\begin{enumerate}
\item If $\core(\fancya)$ has unbounded pathwidth,
then $$\HOMP{\fancya} \equiv_\pl \HOMP{\T^*}.$$

\item If $\core(\fancya)$ has bounded pathwidth and unbounded tree depth,
then $$\HOMP{\fancya} \equiv_\pl \HOMP{\P^*}.$$

\item If $\core(\fancya)$ has bounded tree depth, then 
$$
\HOMP{\fancya}\in\textup{para-L}.$$

\end{enumerate}
\end{theorem}

\begin{remark}\em  If $\fancya$ is assumed to be only computably enumerable instead of decidable, then the theorem stays true 
understanding all mentioned problems in a suitable way as promise
problems. 
If no computability assumption is placed on $\fancya$,
then the theorem stays true in the non-uniform setting of parameterized 
complexity theory (cf.~\cite{downeyfellows}).
\end{remark}

We break the proof into several lemmas. \medskip

%Koalitis and Vardi ~\cite[Lemma~5.2]{koalititsvardicontainment} 
%showed that the canonical query of a structure of {\em treewidth} $w$ is eq%uivalent to a existential first-order sentence that uses only
%$w+1$ {\em variables}; model-checking such sentences can be done in
%poly%omial time~\cite{vardi}. 

To prove statement (3) of Theorem~\ref{thm:classification} we show
that a structure of {\em tree depth} $w$ 
can be characterized, in a sense made precise, by an existential first-order 
sentence of {\em quantifier rank} $w+1$, and that model-checking such sentences can be done in parameterized logarithmic space.  A proof can be found in Section~\ref{sec:fo}.
%This yields statement (3) in Theorem~\ref{thm:classification}.

\begin{lemma}\label{lem:tdinL} Assume $\mathcal A$ is a decidable class of structures of bounded 
arity such that $\core(\fancya)$ has bounded tree depth. Then  $\HOMP{\fancya}\in\textup{para-L}$.
\end{lemma}

To prove statements (1) and (2) of Theorem~\ref{thm:classification}
we need to deal with homomorphism problems for classes $\mathcal A$ that are not 
necessarily decidable. Slightly abusing notation, we say $\HOMP{\fancya}\le_\pl\HOMP{\fancya'}$
for arbitrary classes of structures $\mathcal A,\mathcal A'$ if 
there is a implicitly pl-computable {\em partial} function $F$ that is defined on those 
instances $(\rela,\relb)$ of $\HOMP{\fancya}$ with $\rela\in\fancya$ and maps them to
equivalent instances $(\rela',\relb')$ of $\HOMP{\fancya'}$ with $\rela'\in\fancya'$ such 
that $|\rela'|$ is effectively bounded in $|\rela|$.
By saying that a partial function $F$ is implicitly pl-computable with respect to a parameterization $\kappa$ we mean that there are a computable $f:\mathbb N\to\mathbb N$ and a Turing machine that
on those instances $(x,i,b)$ of $\textsc{Bitgraph}(F)$ such that $F$ is defined on $x$, runs in space $O(f(\kappa(x))+\log|x|)$
and answers $(x,b,i)\stackrel{?}{\in} \textsc{Bitgraph}(F)$; on other instances the machine may do whatever it wants.

%To prove statements (1) and (2) of Theorem~\ref{thm:classification}
%we need to deal with homomorphism problems for classes $\mathcal A$ that are not 
%necessarily decidable. The problem $\HOMP{\fancya}$ is understood to have the parameterization that maps a pair of structures 
%to the size of the first structure.
%For arbitrary classes of structures $\mathcal A,\mathcal A'$ we define 
%$\HOMP{\fancya}\le_\pl\HOMP{\fancya'}$
%to mean that there exists an implicitly pl-computable function (with respect to the parameterization specified above) that 
%maps every instance $(\rela,\relb)$ with $\rela\in\fancya$ of $\HOMP{\fancya}$ to an 
%equivalent instance $(\rela',\relb')$ with $\rela'\in\fancya'$ of $\HOMP{\fancya'}$ such 
%that $|\rela'|$ is effectively bounded in $|\rela|$.\medskip

The following lemma takes care of the reductions from left to right in statements (1) and (2) of Theorem~\ref{thm:classification}. 

\begin{lemma}
\label{lemma:tree-decomp-to-tree}
Let $\fancya$ be a class of structures %of bounded arity 
and $\fancyr \subseteq \T$ be a computably enumerable class of trees.
Assume there is $w \in \nats$ such that every structure in $\fancya$ has a tree decomposition of width 
at most $w$ whose tree is contained in~$\fancyr$.
Then, $$\HOMP{\fancya} \le_{\pl} \HOMP{\fancyr^{*}}.$$
\end{lemma}

\begin{proof}
Let $(\rela, \relb)$ with $\rela\in\fancya$ be an instance of $\HOMP{\fancya}$. Enumerating $\mathcal R$, test successively for $\relt\in\mathcal R$ whether there exists a width $\le w$ tree-decomposition $(\relt,(X_t)_{t \in T})$ of $\rela$.
% with $\relt\in\mathcal R$. 
Since $\rela\in\mathcal A$ this test eventually succeeds, and the time needed is effectively bounded in the parameter $|\rela|$.
 With such a tree-decomposition at hand
 produce the instance $(\relt^*, \relb')$ of the problem $\HOMP{\fancyr^*}$ where the structure $\relb'$ is defined as follows. Write $\dom(f)$ for the domain of a partial function $f$; two partial functions $f$ and $g$ are {\em compatible} if they agree on arguments where they are both defined.
\begin{eqnarray*}
B'&:=& \big\{ f\mid f \mbox{ is a partial homomorphism from } \rela \text{ to } \relb
\text{ and } |\dom(f)| \leq w \big\};\\
E^{\relb'}&:=&\big\{ (f, g) \in B' \times B'\mid \text{ $f$ and $g$ are compatible}\big\};\\
C_t^{\relb'} &:=& \big\{ f\in B'\mid \dom(f) = X_t \big\},\quad\text{ for every $t\in T$}.
\end{eqnarray*}
Suppose that $h$ is a homomorphism from $\rela $ to $\relb$.
Then the mapping $h': T \rightarrow B'$ defined by
$h'(t) = h \upharpoonright X_t$ is straightforwardly verified to be a homomorphism from
$\relt^*$ to $\relb'$.

Conversely, let $h'$ be a homomorphism from $\relt^*$ to $\relb'$. Then, $h'(t)$ is a partial homomorphism from $\rela$ to $\relb$ with domain $X_t$. Since $\relt$ is connected the values of $h'$ are pairwise compatible. Hence 
$h:=\bigcup_{t \in T} h'(t)$ is a function from $\bigcup_{t\in T}X_t=A$ to $B$.
To see $h$ is a homomorphism, consider a tuple 
$(a_1, \ldots, a_r) \in R^{\rela}$ for some $r$-ary relation $R$ in the vocabulary of $\rela$. 
Then $\{ a_1, \ldots, a_r \}$ is contained in some bag $X_t$ since it is a clique in the Gaifman graph of $\rela$ (cf.~\cite[Lemma~4]{bodlaendersurvey}). But $h'(t)$ maps this tuple to a tuple in $R^{\relb}$, so
the mapping $h$ does as well.
\end{proof}

For later use we make the following remark concerning the above proof.

\begin{remark}\label{rem:count} \em 
The previous proof associates with a homomorphism $h$ from $\rela $ to $\relb$
the homomorphism $h'$ from $\relt^*$ to $\relb'$ that maps $t$ to $h\upharpoonright X_t$. This 
association $h\mapsto h'$ is injective because every $a \in A$ appears in some bag $X_t$. It is also surjective: 
a homomorphism $h'$ from $\relt^*$ to $\relb'$, is associated with $h:=\bigcup_{t\in T}h'(t)$; the previous proof argued that $h$ is a homomorphism from $\rela$ to $\relb$.
Hence, there is a bijection between 
the set of homomorphisms from $\rela$ to $\relb$ 
and the set of homomorphisms from $\relt^*$ to $\relb'$.
\end{remark}

At the heart of the proof of Theorem~\ref{thm:classification} is the following sequence of reductions, proved in the following subsection.
The appropriately informed reader will recognize  elements from
Grohe's proof~\cite{Grohe07-otherside} as well as from
Marx~\cite[Lemma 5.2]{marx-toc-treewidth}.

\begin{lemma}[Reduction Lemma] \label{lem:redrow}
Let $\fancya$ be a computably enumerable class of structures of bounded arity, 
let $\fancyg$ be the class of Gaifman
graphs of $\core(\fancya)$, and let $\fancym$ be the class of minors of graphs
in $\fancyg$. Then
\begin{eqnarray*}
\HOMP{\fancym^*} 
&\le_{\pl}& \HOMP{\fancyg^*} \\
&\le_{\pl} &\HOMP{\core(\fancya)^*}\\
&\le_{\pl} &\HOMP{\core(\fancya)} \\
&\le_{\pl}& \HOMP{\fancya}.
\end{eqnarray*}
\end{lemma}

With the Reduction Lemma, we can give the proof of the Classification Theorem.\medskip

\begin{proof}[ Proof of Theorem~\ref{thm:classification}]
The reduction from left to right in statements (1) and (2) follow from 
Lemma~\ref{lemma:tree-decomp-to-tree}.
The reductions from right to left follow from the Reduction Lemma~\ref{lem:redrow} via the Excluded Tree Theorem~\ref{thm:rs}~(2)
and the Excluded Path Theorem~\ref{thm:rs}~(3). Statement (3) is proved as Lemma~\ref{lem:tdinL}.
\end{proof}

\subsection{Proof of the Reduction Lemma}

As a consequence of the assumption that $\fancya$ is computably enumerable,
each of the sets $\fancym^*$, $\fancyg^*$, $\core(\fancya)^*$,
and $\core(\fancya)$ are computably enumerable. The statement of the theorem claims the existence of four reductions. The last one from $ \HOMP{\core(\fancya)} $
to $ \HOMP{\fancya}$ is easy to see. We construct the first three in sequence.

\begin{lemma}
\label{lemma:red-minors-to-graphs}
Let $\fancyg$ be a %computably enumerable 
class of graphs which is computable enumerable,
and let $\fancym$ be the class of minors of graphs in~$\fancyg$.
Then $$\HOMP{\fancym^*}  \le_{\pl} \HOMP{\fancyg^*}.$$
\end{lemma}

\begin{proof}
Let $(\relm^*, \relb)$ with $\relm^*\in\fancym^*$ be an instance of the problem $\HOMP{\fancym^*}$. Enumerating $\fancyg$, test successively for $\relg\in\fancyg$ whether $\relm$ is a minor of $\relg$. Since $\relm\in\fancym$ this test eventually succeeds, and then compute a minor map $\mu$ from $\relm$ to $\relg$. 
The time needed is effectively bounded in the parameter $|\relm^*|$. 
The reduction then produces the instance 
$(\relg^*, \relb')$ of $\HOMP{\fancyg^*}$, where $\relb'$ is defined as follows. 
Let $I$ denote the set $\bigcup_{m \in M} \mu(m)$.
\begin{eqnarray*}
B'&:=& (M \times B) \dot\cup \{ \bot \};\\
E^{\relb'}&:=& \big\{ ((m_1, b_1), (m_2, b_2))\mid  [m_1 = m_2 \Rightarrow b_1 = b_2] \mbox{ and }\\
&&\qquad\qquad\qquad\qquad\qquad\qquad [(m_1, m_2) \in E^{\relm} \Rightarrow (b_1, b_2) \in E^{\relb}] \big\}\\
&&\cup\ \big\{(\bot,b')\mid b'\in B' \}\cup \{(b',\bot)\mid b'\in B'\big\};\\
C_v^{\relb'}&:=&\{(m,b)\mid b\in C_m^{\relb}\},\quad\text{ if }m\in M\text{ and }v\in\mu(m);\\
C_v^{\relb'}&:=&\{\bot\},\quad\text{ if }v\notin I.
\end{eqnarray*}

Suppose that $h$ is a homomorphism from
$\relm^*$ to $\relb$.  
Let $h': G \rightarrow B'$ be the map that sends, for each $m \in M$,
the elements in $\mu(m)$ to $(m, h(m))$ and that sends
all elements $v \notin I$ to $\bot$.
Then $h'$ is
a homomorphism from $\relg^*$ to~$\relb'$.

Suppose that $g$ is a homomorphism from $\relg^*$ to $\relb'$.  We show that $g$ is of the form $h'$
for a homomorphism $h$ from $\relm^*$ to $\relb$.
First, by definition of the $C_v^{\relb'}$, it holds that $g(v) = \bot$ for all $v \notin I$.
Next, let $v, w$ be elements of a set $\mu(m)$, with $m \in M$.
The definition of the $C_v^{\relb'}$ ensures that $g(v)$ and $g(w)$ have the form~$(m, \cdot)$. Since $\mu(m)$ is connected, the  definition of $E^{\relb'}$ ensures that $g(v) = g(w)$. Finally, suppose that $(m_1, m_2) \in E^{\relm}$,
let $(m_1, b_1)$ be the image of $\mu(m_1)$ under $g$, and let $(m_2, b_2)$ be the image of $\mu(m_2)$ under~$g$. We claim that $(b_1, b_2) \in E^{\relb}$. But there  exist $v_1 \in \mu(m_1)$ and $v_2 \in \mu(m_2)$ such that $(v_1, v_2) \in E^{\relg}$. We then have $(g(v_1), g(v_2)) \in E^{\relb'}$ and the definition
of $E^{\relb'}$ ensures that $(b_1, b_2) \in E^{\relb}$.
\end{proof}

\begin{lemma}
\label{lemma:red-graphs-to-structures}
Let $\fancya$ be a computably enumerable class of structures of bounded arity, 
and let $\fancyg$ be the class of Gaifman graphs of~$\fancya$. Then $$\HOMP{\fancyg^*}  \le_{\pl} \HOMP{\fancya^*}.$$
\end{lemma}

\begin{proof}
%Let the pair $(\relg^*, \relb)$ with $\relg\in \fancyg$ be an instance of $\HOMP{\fancyg^*}$. Similarly as seen in the previous proof, one can compute from $\relg$ a structure $\rela\in\fancya$ whose Gaifman graph is $\relg$.
%The reduction produces the instance $(\rela^*, \relb')$ of
%$\HOMP{\fancya^*}$, where $\relb'$ is defined as follows.
%\begin{eqnarray*}\itemsep=0pt
%B'&:=&B;\\
%C_a^{\relb'}&:=&C_a^{\relb},\quad \text{ for every } a\in A;\\
%R^{\relb'}&:=&\{(b_1,\ldots,b_r)\mid\{b_1\ldots,b_r\}\text{ is a clique in }(B,E^{\relb})\},
%\end{eqnarray*}
%for every $r$-ary symbol $R$ from the vocabulary of $\rela$. It is straightforward to verify that a function
%$h: A \rightarrow B$ is a homomorphism from $\relg$ to $\relb$ if and only if
%it is a homomorphism from $\rela$ to $\relb'$.
%

Let $(\relg^*, \relb)$ with $\relg\in \fancyg$ be an instance of $\HOMP{\fancyg^*}$. 
Similarly as seen in the previous proof, one can compute from $\relg$ a structure $\rela\in\fancya$ whose Gaifman graph is $\relg$;
 in particular, $A=G$ and we write $\relg=(A,E^{\relg})$. The reduction outputs $(\rela^*,\relb')$ where $\relb'$ is the structure 
defined as follows.
\begin{eqnarray*}
B'&:=&A\times B,\\
C^{\relb'}_a&:=& \{a\}\times C_a^{\relb}, \\
R^{\relb'}&:=&\Big\{((a_1,b_1),\ldots(a_{\ar(R)},b_{\ar(R)}))\in (A\times B)^{\ar(R)}\mid \\
&&\qquad \bar a\in R^{\rela}\textup{ and for all } i,j\in[\ar(R)]:\textup{ if } a_i\neq a_j, \textup{ then } (b_i,b_j)\in E^{\relb}   \Big\},
\end{eqnarray*}
for $R\in\tau$ where $\tau$ denotes the vocabulary of $\rela$. We have to show
\begin{equation*}
 (\relg^*,\relb)\in\HOMP{\fancyg^*}\Longleftrightarrow (\rela^*,\relb')\in\HOMP{\fancya^*}.
\end{equation*}

To see this, assume first that $h$ is a homomorphism from $\relg^*$ to $\relb$. We claim that $h'(a):=(a,h(a))$ 
defines a homomorphism from $\rela^*$ to $\relb'$. If $a'\in C_a^{\rela^*}$, then $a'=a$ and $h(a')\in C_a^{\relb}$ since $h$ is a homomorphism; 
by definition then $h'(a')=(a,h(a))\in C_a^{\relb'}$. Hence $h'$ preserves the symbols $C_a$. To show it preserves $R\in \tau$, let
$(a_1,\ldots,a_{\ar(R)})\in R^{\rela}$. We have to show $((a_1,h(a_1)),\ldots,(a_{\ar(R)},h(a_{\ar(R)})))\in R^{\relb'}$, or equivalently, for all 
$i,j\in[\ar(R)]$ with $a_i\neq a_j$ that $(h(a_i),h(a_j))\in E^{\relb}$. But if $a_i\neq a_j$, then $(a_i,a_j)\in E^{\relg}$ by definition of the 
Gaifman graph and $(h(a_i),h(a_j))\in E^{\relb}$ follows from $h$ being a homomorphism.

Conversely, assume that $h'$ is a homomorphism from $\rela^*$ to $\relb'$. By definition of $C_a^{\relb'}$ is follows that $h'(a)=(a,h(a))$ for some 
function $h:A\to B$ such that $h(a)\in C_a^{\relb}$. We claim that $h$ is a homomorphism from $\relg^*$ to $\relb$. It suffices to show
$(h(a),h(a'))\in E^{\relb}$ whenever $(a,a')\in E^\relg$. But if $(a,a')\in E^\relg$, then $a\neq a'$ and there exist $R\in\sigma$ and 
$(a_1,\ldots,a_{\ar(R)})\in R^{\rela}$ and $i,j\in[\ar(R)]$ such that $a=a_i$ and $a'=a_j$. Then 
$((a_1,h(a_1)),\ldots,(a_{\ar(R)},h(a_{\ar(R)})))\in R^{\relb'}$ because $h'$ is a homomorphism. Since $a_i\neq a_j$ the definition of $E^{\relb'}$ implies
$(h(a_i),h(a_j))=(h(a),h(a'))\in E^{\relb}$ as desired.
\end{proof}

Recall that the {\em direct product} $\rela\times\relb$ of two $\tau$-structures $\rela$ and~$\relb$ has
universe $A\times B$ and interprets a relation symbol $R\in\tau$ by $\{((a_1,b_1),\ldots,(a_{\ar(R)},b_{\ar(R)}))\mid \bar a\in R^\rela,\bar b\in R^\relb\}$.

\begin{lemma}
\label{lemma:reduction-constants-cores}
Let $\fancya$ be a class of structures.
Then $$\HOMP{\core(\fancya)^*} \le_{\pl} \HOMP{\core(\fancya)}.$$
\end{lemma}

\begin{proof} 
Let $(\reld^*, \relb)$ with $\reld\in\core(\fancya)$ be an instance of $\HOMP{\core(\fancya)^*}$.
Let $\relb_*$ be the restriction of $\relb$ to the vocabulary of $\reld$.
The reduction produces the instance $(\reld, \relb')$
of the problem $\HOMP{\core(\fancya))}$, where 
$$
\relb':=\big\langle \big\{ (d, b) \in D \times B\mid b \in C_{d}^{\relb}  \big\}  \big\rangle^{\reld \times \relb_*}.
$$ 

Suppose that $h$ is a homomorphism from $\reld^*$ to $\relb$.
Then, the mapping $h': D \rightarrow B'$ defined by
$h'(d) = (d, h(d))$ is straightforwardly verified to be a homomorphism
from $\reld$ to $\relb'$.

Suppose that $g$ is a homomorphism from $\reld$ to $\relb'$. Write $\pi_1$ and $\pi_2$ for the projections that map a pair to its first and second component respectively.
The composition $(\pi_1 \circ g)$ is a homomorphism from $\reld$ to itself; since $\reld$ is a core,
$(\pi_1 \circ g)$ is bijective.  Hence, there exists a natural $m \geq 1$ such that
$(\pi_1 \circ g)^m$ is the identity  on~$D$. Define $h$ as 
$g \circ (\pi_1 \circ g)^{m-1} $. Clearly, $h$ is a homomorphism from 
$\reld$ to $\relb'$, so $\pi_2 \circ h$ is a homomorphism from 
$\reld$ to $\relb_*$.  We claim that $\pi_2 \circ h$ is also a homomorphism from $\reld^*$ to~$\relb$.
Observe that $\pi_1 \circ h$ is the identity on $D$. In other words, for every $d\in D$ there is $b_d \in B$ such that
$h(d)=(d, b_d)$. By definition of $\relb'$ we get $b_d \in C_d^{\relb}$,
establishing the claim.
\end{proof}

Observe that the map $h'$ constructed in the above proof is an embedding. Hence
we have the following corollary that we note explicitly for later use.

\begin{corollary}\label{cor:coreemb}
Let $\fancya$ be a class of structures.
Then $$
\HOMP{\core(\fancya)^*} \le_{\pl} \EMP{\core(\fancya)}.
$$
\end{corollary}

\subsection{Bounded tree depth and $\textup{para-L}$}\label{sec:fo}

Let $\tau$ be a vocabulary. {\em First-order $\tau$-formulas} are built from {\em atoms} $R\bar x, x=x$ by Boolean 
combinations and existential and universal quantification. Here, $\bar x$ is a tuple of variables of length matching the arity of $R$. 
We write $\varphi(\bar x)$ for a (first-order) $\tau$-for\-mula~$\varphi$
to indicate that the free 
variables in $\varphi$ are among the components of $\bar x$. The {\em quantifier rank} $\qr(\varphi)$ of a formula $\varphi$
is defined as follows: 
\begin{itemize}\itemsep=0pt
\item[] $\qr(\varphi)=0$ \quad for atoms $\varphi$; 
\item[] $\qr(\neg\varphi)=\qr(\varphi)$; 
\item[] $\qr(\varphi\wedge\psi)=\qr(\varphi\vee\psi)=\max\{\qr(\varphi),\qr(\psi)\}$; 
\item[] $\qr(\exists x\varphi)=\qr(\forall x\varphi)=1+\qr(\varphi)$.
\end{itemize}

The following is standard, but we could not find a reference, so include the simple proof for completeness.

\begin{lemma}\label{lem:mcfo} The parameterized problem
\npprob{$p$-MC(FO)}{\em A structure $\rela$, a first-order sentence $\varphi$}
{$|\varphi|$}{$\rela\models\varphi$ ?}
can be decided in space  $O(|\varphi|\cdot\log|\varphi|+(\qr(\varphi)+\ar(\varphi))\cdot\log |A|)$, where 
$\qr(\varphi)$ is the quantifier rank of $\varphi$ and 
$\ar(\varphi)$ is the maximal arity over all relation symbols in $\varphi$
\end{lemma}

\begin{proof} We give an algorithm expecting inputs $(\rela,\varphi,\alpha)$ where 
$\varphi$ is a formula and 
$\alpha$ is an assignment for $\varphi$ in $\rela$, that is, a map from a superset of the free variables of $\varphi$ into $A$. The algorithm determines whether $\alpha$ satisfies $\varphi$ in $\rela$. It executes a depth-first recursion as follows. 

If $\varphi$ is an atom $R\bar y$ the algorithm writes 
the tuple $\alpha(\bar y)\in A^{\ar(R)}$ on the worktape and checks whether it is contained in $R^\rela$ by scanning the input; it then erases the tuple and returns the bit corresponding to the answer obtained. 

If $\varphi=(\psi\wedge\chi)$, the algorithm recurses on $\psi$ (with the same assignment); 
upon completing the recursion it erases all space used in it, stores a bit for the answer obtained, and then recurses on $\chi$;
upon completion it erases the space used in it and returns the minimum of the bit obtained and the stored bit.
The cases $\varphi=(\psi\vee\chi)$ and $\varphi=\neg\psi$ are similar.

If $\varphi(\bar x)=\exists y\psi(\bar x,y)$ the algorithm loops through $b\in A$ and recurses on $\psi$ with assignment $\alpha$ 
extended by mapping $y$ to $b$; 
it maintains a bit which is intially 0 and updates it
after each loop to the maximum of the bit obtained in the loop; after each loop it erases the space used in in it.
Upon completing the loop it returns this bit, and restricts the assignment back to its old domain without~$y$.
The case $\varphi(\bar x)=\forall y\psi(\bar x,y)$ is similar.

When started on a sentence $\varphi$ and the empty assignment, all assignments $\alpha$ occuring in the recursion have cardinality $\le \qr(\varphi)$, so can be stored in space $O(\qr(\varphi)\cdot(\log|\varphi|+\log|A|))$. Each recursive step adds space $O(\log|\varphi|)$ to remember
 the (position of) the current subformula plus one bit 
 plus $O(\log |A|)$ 
for the loop on $b\in A$ in the quantifier case and plus $O(\ar(\varphi)\cdot\log |A|)$ in the atomic case.
From these considerations it is routine to verify the claimed upper bound on space.
%
%Each recursive step adds space $O(\log|\varphi|)$ to remember the (position of) the subformula $\psi$ plus one bit plus $O(\log |A|)$ 
%for the loop on $b\in A$ in the quantifier case and plus $O(\ar(\varphi)\cdot\log |A|)$ in the atomic case. The assignments $\alpha$ occuring in the recursion 
%all have cardinality $w\le \qr(\varphi)$, so need space $O(w(\log|\varphi|+\log|A|))$ to be stored. 
\end{proof}

The {\em canonical conjunction} of a structure $\rela$ is a quantifier-free conjunction in the variables $x_a$ for $a\in A$; namely, 
for every relation symbol $R$ of $\rela$ and every  $(a_1,\ldots, a_{\ar(R)})\in R^\rela$ it contains the conjunct
$Rx_{a_1}\cdots x_{a_{\ar(R)}}$. It is easy to see that the canonical
conjunction of $\rela$ 
is satisfiable in a
structure $\relb$ if and only if there is an homomorphism from $\rela$ to~$\relb$.

\begin{proof}[Proof of Lemma~\ref{lem:tdinL}] 
Choose $w\in\mathbb N$ such that $\td(\core(\rela))\le w$ for all $\rela\in\fancya$. 
Given a structure $\rela$ we compute a sentence $\varphi_\rela$ of quantifier rank at most $w+1$ such that for all structures $\relb$, the sentence $\varphi_\rela$ is true in $\relb$ if and only if there
is a homomorphism from $\rela$ to $\relb$. 
This is enough by Lemma~\ref{lem:mcfo}.

Given $\rela$ we check $\rela\in\fancya$ running some decision procedure for $\fancya$. If
 $\mathbf A\notin\fancya$ we let $\varphi_\rela:=\exists x\ \neg x=x$. 
If $\rela\in\fancya$, compute the core $\mathbf A_0$ of $\mathbf A$ and compute for every connected component $C$ of the Gaifman graph of $\rela_0$ some 
rooted tree $\relt$ with vertices $T=C$ and height at most $w$ such that every edge of the 
Gaifman graph of $\langle C\rangle^{\rela_0}$ is in the closure of $\relt$.

Consider a component $C$ and let $\relt$ be the rooted tree computed for $C$.
For $c\in C=T$ we compute the following first-order formula $\varphi_c$. We use variables $x_c$ for $c\in C=T$. 
If $c$ is a leaf of $\relt$, let $\varphi_c$ be the canonical conjunction of $\langle P_c\rangle^{\rela_0}$ where $P_c$ is 
the path  in $\relt$ leading from the root $r$ of $\relt$ to $c$.
For an inner vertex $c$ define
$$\textstyle
\varphi_c:=\bigwedge_{d} \exists x_d\ \varphi_d,
$$ 
where $d$ ranges over the successors of $c$. 
The following claims are straightforwardly verified by induction along the recursive definition of the $\varphi_c$s.

\medskip

\noindent{\em Claims.} For every $c\in C$:
\begin{enumerate}\itemsep=0pt
 \item the quantifier rank of $\varphi_c$ equals the height of the subtree of $\relt$ rooted at $c$;
\item the free variables of $\varphi_c$ are $\{x_d\mid d\in P_c\}$;
\item $\varphi_c$ is satisfiable in $\relb$ if and only if so is the canonical conjunction of 
$\langle C(c)\rangle^{\rela_0}$ where $C(c)$ contains $P_c$ and the vertices in the subtree rooted at $c$.
\end{enumerate}

Letting $r$ range over the roots of the trees $\relt$ chosen for the connected components $C$ of $\rela_0$, we set
$$\textstyle
\varphi_\rela:=\bigwedge_r\exists x_r\varphi_r.
$$ 
By Claim 2 this is a sentence and by Claim 1 it has quantifier rank at most $w+1$. It is true in $\relb$ if and only if every 
$\exists x_r\varphi_r$ is true in $\relb$, and by Claim 3 this holds  if and 
only if the canonical conjunction of 
$\langle C(r)\rangle^{\rela_0}$ is satisfiable in $\relb$ for every 
connected component $C$. Noting $C(r)=C$, this means
that every $\langle C\rangle^{\rela_0}$ maps 
homomorphically to $\relb$, and this means that $\rela_0$ maps homomorphically to $\relb$. Recalling that
$\rela_0$ is the core of $\rela$, we see that this is equivalent to $\rela$ mapping homomorphically to $\relb$.
 \end{proof}

Define a
$\{ \wedge, \exists \}$-sentence
to be a first-order sentence built from atoms,
conjunction,
and existential quantification.
The previous proof revealed that, given a structure $\rela$
with 
$\td(\core(\rela))\le w$,
there exists a 
$\{ \wedge, \exists \}$-sentence $\phi$ 
of quantifier rank at most $w + 1$
that \emph{corresponds} to $\rela$ in that,
for all structures $\relb$,
the sentence $\phi$ is true on $\relb$ 
if and only if there is a homomorphism
from $\rela$ to $\relb$.
We show that the existence of such a sentence in fact characterizes
tree depth, in the following precise sense.

\begin{theorem}
Let $w \geq 0$, and let $\rela$ be a structure.
It holds that $\td(\core(\rela)) \le w$
if and only if
there exists a $\{ \wedge, \exists \}$-sentence $\phi$
that corresponds to $\rela$ with $\qr(\phi) \leq w + 1$.
\end{theorem}

\begin{proof}
The forward direction follows from the previous proof.
For the backward direction, let $\phi$ be a sentence of the described
type.
We may assume that no variable is quantified twice in $\phi$
and that no equality of variables appears in $\phi$,
by renaming variables and replacing equalities of the form $v = v$
with the empty conjunction.
Let $\phi_p$ be the prenex sentence where all variables that are
existentially quantified in $\phi$ are existentially quantified in
$\phi_p$,
and the quantifier-free part of $\phi_p$ is the conjunction of all 
atoms appearing in $\phi$.
Let $\relc$ be a structure whose canonical conjunction is 
the quantifier-free part of $\phi_p$.
Clearly, $\phi_p$ and the original $\phi$ are logically equivalent;
it follows that $\relc$ and $\rela$ are 
homomorphically equivalent~\cite{ChandraMerlin77-optimal}.
It thus suffices to show that $\td(\relc) \leq w$.

View the sentence $\phi$ as a directed graph, and
define an acyclic directed graph $D$ on the variables of $\phi$
where the directed edge $(v, v')$ is present if and only if
the node for $\exists v$ is the first node with quantification
occurring above the node for $\exists v'$.
Let $\alpha$ be an arbitrary atom from $\phi_p$
(equivalently, from $\phi$).
Since $\phi$ is a sentence, if one traverses $\phi$ starting from 
the root and moving to $\alpha$,
one will pass a node $\exists v$ for each variable $v$ of
$\alpha$.
Let $v_1, \ldots, v_k$ be the variables of $\alpha$ in the order
encountered by such a traversal.  
The edges $(v_1, v_2), (v_2, v_3), \ldots, (v_{k-1}, v_k)$ are in 
the transitive closure of $D$, and hence in the closure of 
the graph underlying $D$ (where a node is a root in the graph 
iff it is parentless in $D$).
Since $\qr(\phi) \leq w + 1$, each directed path in $D$ has length
less than or equal to $w$, and so 
the graph underlying $D$ witnesses that $\td(\relc) \leq w$.
\end{proof}

We now show the following result on the embedding problem.

\begin{theorem}\label{theo:embtd}
 Assume $\mathcal A$ is a decidable class of structures of bounded 
arity and  bounded tree depth. Then  $\EMP{\fancya}\in\textup{para-L}.$
\end{theorem}

The proof of this result uses color coding methods, more precisely, it relies on the following lemma (see~\cite[p.349]{flumgrohe}).

\begin{lemma}\label{lem:hash} For every sufficiently large $n$, it holds that for all $k\in\mathbb N$ and 
for every $k$-element subset $X$ of $[n]$, 
there exists a prime $p<k^2\log n$ and $q<p$ such that the function 
$h_{p,q}:[n]\to\{0,\ldots,k^2-1\}$ given by 
$$
h_{p,q}(m):= (q\cdot m \textup{ mod } p) \textup{ mod } k^2
$$
 is injective on $X$. 
\end{lemma}

For later use we give the main step in the proof of Theorem~\ref{theo:embtd} as a separate lemma. Call a structure {\em connected} if 
its Gaifman graph is connected.

\begin{lemma}\label{lem:connected}
For every decidable class of connected structures $\fancya$ we have $$\EMP{\fancya}\le_\pl\HOMP{\fancya^*}.$$
\end{lemma}

\begin{proof}
 Map an instance $(\rela,\relb)$ to $(\rela^*,\relb_*)$ where $\relb_*$ is defined as follows. We assume that $B=[|B|]$ and $A=[|A|]$. Let $F$ be the set
$$
\big\{g\circ h_{p,q}\mid  g:\{0,\ldots,|A|^2-1\}\to A \text{ and }q<p<|A|^2\log|B|\big\}.
$$
Here, $h_{p,q}:[|B|]\to\{0,\ldots,|A|^2-1\}$ is the function from Lemma~\ref{lem:hash} (for $n:=|B|$ and $k:=|A|$).
For $f\in F$, let $\relb_f$ be the expansion of $\relb$ that interprets every $C_a,a\in A,$ by $f^{-1}(a)\subseteq B$ and define $\relb_*$ as the disjoint union of the structures $\relb_f$. 
We verify 
$$
(\rela,\relb)\in\EMP{\fancya}\Longleftrightarrow (\rela^*,\relb_*)\in\HOMP{\fancya^*}.
$$
Note that the sets $C^{\relb_*}_a, a\in A$,  are pairwise disjoint, so every homomorphism from $\rela^*$ to $\relb_*$ is an embedding. And because $\rela^*$ is connected, it is an embedding into (the copy of) some $\relb_f$, so it corresponds to an embedding from $\rela$ into~$\relb$.
 Conversely, assume $e$ is an embedding of $\rela$ into $\relb$. 
By Lemma~\ref{lem:hash} there are $p,q$ with $q<p<|A|^2\log |B|$ such that $h_{p,q}$ is 
injective on the image of $e$. Then there exists $g:\{0,\ldots,|A|^2-1\}\to A$ 
such that $g\circ h_{p,q}\circ e$ is 
the identity on $A$. Then $f:=g\circ h_{p,q}\in F$ and $e$ is an embedding of $\rela^*$ into~$\relb_{f}$ and hence into $\relb_*$. 
\end{proof}

This lemma together with Corollary~\ref{cor:coreemb} implies:

\begin{corollary}\label{cor:embhom}
Let $\fancya$ be a decidable class of connected cores. Then $$\HOMP{\fancya^*}\equiv_\pl\EMP{\fancya}.$$
\end{corollary}

\begin{proof}[Proof of Theorem~\ref{theo:embtd}] Let $\fancya$ accord the assumption. 

\medskip

\noindent{\em Claim.} There exists a decidable class of connected structures $\fancya'$ of
 bounded tree depth such that $\EMP{\fancya}\le_\pl\EMP{\fancya'}$.\medskip

Note $\EMP{\fancya'}\le_\pl\HOMP{(\fancya')^*}$ by the previous lemma and $\HOMP{(\fancya')^*}\in\textup{para-L}$ by Lemma~\ref{lem:tdinL}. 
We are thus left to prove the claim.

 Assume $\fancya$ has tree depth at most $d$ and let $E$ be a binary relation symbol not occuring in the vocabulary 
of any $\rela\in\fancya$. Fix a computable function that maps every $\rela\in\fancya$ to a family
of height $\le d$ rooted trees $(\relt_C)_C$ with $\relt_C=(C,E^{\relt_C},\textit{root}^{\relt_C})$ where $C$
ranges over the connected components of the 
Gaifman graph $\relg(\rela)$ of $\rela$, and such that $\langle C\rangle^{\relg(\rela)}$ is a subgraph of the closure of $\relt_C$. 
Define $\rela'$ to be the expansion of $\rela$ interpreting $E$ by $\bigcup_CE^{\relt_C}\cup E'$ where $E'$ is defined as follows. It
contains edges between the root of $\relt_{C_0}$ and the roots of the other $\relt_{C}$ 
where $C_0$ is the lexicographically minimal component (according to the encoding of $\rela$).
Then $\rela'$ is connected and has tree depth at most $d+1$. 
Clearly, $\fancya':=\{\rela'\mid\rela\in\fancya\}$ is decidable. The map $(\rela,\relb)\mapsto(\rela',\relb')$, where $\relb'$ 
is the expansion of $\relb$ interpreting $E$ by $B^2$, is a pl-reduction from $\EMP{\fancya}$ to $\EMP{\fancya'}$.
\end{proof}

\section{The class PATH}
\label{sect:path}

\newcommand{\pstpath}{\textsc{$p$-$st$-Path}}

We present the complexity class $\PATH$ to capture the complexity of $\HOMP{\P^*}$. 
This class was discovered very recently by  Elberfeld et al.~\cite{tantau} with a different angle of motivation;
they refer to this class as para-NL[$f \log$]. Among other results, they show that the following problem is
complete for this class: check if a digraph contains
a %(not necessarily simple) 
path from a distinguished vertex $s$ to another distinguished vertex $t$
of length at most $k$; here, $k$ is the parameter.
We use
$\pstpath$ to denote the corresponding problem for undirected graphs.

\npprob{$\pstpath$}{A graph $\relg$, $s,t\in G$ and $k\in\nats$}{$k$}{Is there a path in $\relg$ from $s$ to $t$ of length at most~$k$~?}

\begin{definition}\label{def:pathmachine} 
The class $\PATH$ contains a
parameterized problem $(Q,\kappa)$ if there are
a computable function $f:\mathbb N\to\mathbb N$ 
and a nondeterministic Turing machine that 
accepts $Q$, is pl-space boun\-ded with respect to $\kappa$, and uses $f(\kappa(x))\cdot\log |x|$ many nondeterministic bits.
\end{definition}

The following is straightforward to verify.

\begin{prop}
The complexity class $\PATH$ is closed under pl-reductions.
\end{prop}

Recall that, using the notation in~\cite{describing}, one has 
$$
\textup{FPT}=\textup{para-P}\subseteq\textup{W[P]}\subseteq\textup{para-NP}.
$$ 
It follows  immediately from the definitions that
$$
\textup{para-L}\subseteq\PATH\subseteq\textup{para-NL}.
$$
The class $\PATH$ is natural in that it has a natural machine characterization
that is analogous to the one of W[P].  
We shall see that it captures the complexity of many natural problems.

\begin{theorem}\label{theo:pathcomplete} $\HOMP{\P^*}$ is complete for $\PATH$ under pl-reductions.
\end{theorem}

%\npprob{$p$-DirPath}{a directed graph $\relg$ and $k\in\mathbb %N$}{$k$}{does $\relg$ contain a path of length $k$?}
%
%\npprob{$p$-DirCycle}{a directed graph $\relg$ and $k\in\mathbb %N$}{$k$}{does $\relg$ contain a cycle of length $k$?}
%
%
%\npprob{$p$-Cycle}{a graph $\relg$ and $k\in\mathbb N$}{$k$}{does %$\relg$ contain a cycle of length $k$?}
%

That $\HOMP{\P^*}$ is contained in $\PATH$ 
can be seen by the guess-and-check para\-digm. We find it informative 
to present such algorithms in a computational model tailored specifically for this kind of nondeterminism.

%namely {\em Turing machines with jumps}. Theorem~\ref{theo:pathemb} %can be shown by color-coding methods %(cf.~\cite[Chapter~XXX]{flumgrohe}). In fact, these methods can be used to %show more: one can convert any Turing machine with jumps to one using %{\em injective jumps}. 

%We start by giving the announced machine characterization of $\PATH$ by %Turing machines with jumps and injective jumps and show how to derive %the above results from this characterization.

\begin{definition} 
A \emph{jump machine}
is a Turing machine with an input tape and a special {\em jump state}. 
When the machine enters the jump state the head on the input tape is set nondeterministically on one of the cells carrying an input bit;
we say that the machine {\em jumps to} the cell.
When this occurs,
no other head moves or writes and the state is changed to the starting state. 
Acceptance is defined as usual, that is, such a machine accepts an input
if there exists a sequence of nondeterministic jump choices under which
the machine accepts. 
%It accepts a problem $Q$ if and if it accepts exactly those inputs that are in %$Q$
An \emph{injective jump machine}
is defined similarly to a jump machine,
but never jumps to a cell that has already been jumped to.

For a function $j:\{0,1\}^*\to\mathbb N$, we say that a jump machine 
(an injective jump machine) 
uses
{\em $j$ many} (injective) jumps if
for every input $x$ and every run on $x$, 
it enters the jump state at most $j(x)$ many times.
\end{definition}

The idea is that a jump corresponds to a guess of a number in $[n]$ where $n$ is the length of the input. 
Observe that one can compute in logarithmic space the number $m\in[n]$ of the cell it jumps to by moving the head to the left and stepwise increasing 
a counter. 

\begin{lemma}\label{lem:jumps} Let $(Q,\kappa)$ be a parameterized problem. The following are equivalent.
\begin{enumerate}\itemsep=0pt

\item $(Q,\kappa)\in\PATH$.

\item There exists a computable $f:\mathbb N\to\mathbb N$ and a 
jump machine $\A$ using $(f\circ \kappa)$ many jumps that accepts $Q$ and is pl-space bounded with respect to $\kappa$.

\item There exists  a computable $f:\mathbb N\to\mathbb N$ and 
an injective jump machine $\A$ 
using $(f\circ \kappa)$ many injective jumps that accepts $Q$ and is pl-space bounded with respect to $\kappa$.

\end{enumerate}
\end{lemma}

\begin{proof}   (1) implies (2): assume (1) and choose $\mathbb A$ and $f$ according Definition~\ref{def:pathmachine}.  
Given an input $x$ we simulate $\mathbb A$ by a 
jump machine $\mathbb B$ that makes use of an extra worktape. 
When $\A$ enters its guess state $\B$ moves its head on the extra worktape right and continues the simulation of $\mathbb A$ in state $s_b$ where $b\in\{0,1\}$ is the bit scanned by this head. In case the head scans a blank cell, $\B$ stores the number $j$ of the cell its input head is scanning and then performs a jump, say to cell $m\in[|x|]$. It computes the binary code  of $m$ of length $\lceil\log (|x|+1)\rceil$. It overwrites the content of the extra worktape by this 
code and sets its head on the first bit $b$ of the code, moves the input head back to cell $j$ and continues the simulation of $\A$ in state~$s_b$. 
Then $\mathbb B$ makes at most 
%$f(\kappa(x))\cdot\frac{\log |x|}{\lceil\log (|x|+1)\rceil}\le f(\kappa(x))$ many
$f(\kappa(x))$ many jumps.

(2) implies (3):  let $\mathbb A$ and $f$ accord (2). To get a machine according 
to (3) we intend to simply simulate $\A$ on 
an injective jump machine. 
This works provided $\A$ does not have accepting runs with two jumps to the same cell. To ensure this condition we replace $\A$ by the following machine~$\A'$.
Intutively, if $\A$ jumps $k$ times then $\A'$ jumps $2k$ times and accepts only if these $2k$ jumps encode pairs $(1,m_1),\ldots,(2k,m_{2k})$; the simulation of the $i$th jump of $\A$ is done by jumping to the $(m_{2i},m_{2i+1})$th cell. Details follow.

The machine $\A'$ on $x$ first computes $k:=f(\kappa(x))$: note $\kappa(x)$ can be computed in space $O(\log|x|)$ by our convention on parameterizations; then $k$ can be computed from $\kappa(x)$ running some machine computing $f$ on $\kappa(x)$ -- this needs additional space which is effectively bounded in the parameter $\kappa(x)$. 

Then $\A'$ checks that $2k\cdot\lceil\sqrt{n}\rceil\le n$ where $n:=|x|$. 
If this check fails, $\A'$ simulates some fixed decision procedure for $Q$ (note that  (2) implies that $Q$ is decidable). Observe that in this case $k\ge\Omega(\sqrt{n})$, so the decision procedure runs in space effectively bounded in $k$ and hence in the parameter.
Otherwise $2k\cdot\lceil\sqrt{n}\rceil\le n$ and $\A'$ simulates $\A$ as follows. Throughout the simulation it maintains a counter for jumps that initially is set to 0. It will be clear that this counter always stores a number $\le 2k$.

 When $\A$ jumps, $\A'$ jumps twice and 
computes the two numbers $a,b$ of the cells it jumped to. It interprets $a,b$ as encoding pairs $(i_a,m_a),(i_b,m_b)\in[2k]\times [\lceil\sqrt{n}\rceil]$. More precisely, $i_a:=\lceil a/\lceil\sqrt{n}\rceil \rceil$ is the least $i$ such that $i\cdot\lceil\sqrt{n}\rceil\ge a$ and $m_a:=1+a-(i_a-1)\cdot\lceil\sqrt{n}\rceil$; similarly for $(i_b,m_b)$. If $(i_a,m_a)$ or $(i_b,m_b)$ is not in $[2k]\times [\lceil\sqrt{n}\rceil]$, then $\A'$ halts and rejects.

For $i$ the value of the jump counter, $\A'$ checks that $i+1=i_a$ and that $i+2=i_b$. Then it computes $m:=m_a \cdot\lceil\sqrt{n}\rceil+m_b$ and checks that $m\in[n]$. 
Then $\A'$ increases the jump counter by two, moves the input head to cell $m$, changes to the starting state and resumes the simulation of $\A$.

(3) implies (1): choose a machine $\A$ and a function $f$ according (3) and define a machine $\B$ as follows. On $x$ it first computes $k:=f(\kappa(x))$ (within allowed space as seen above) and $n:=|x|$. If $k\ge \log n$ it runs some fixed machine $\mathbb Q$ deciding $Q$ and answers accordingly. Since $k\ge\log n$ this needs space effectively bounded in $k$ and thus in the parameter.
If otherwise $k<\log n$, then $\B$ simulates $\A$ as follows. During the simulation it maintains 
a set $X$ containing at most $k$ natural numbers  all smaller $k^2$ -- intuitively, this set contains fingerprints of the jumps sofar. Initially, $X=\emptyset$.

To begin, $\B$ guesses a pair $(p,q)$ with $q<p<k^2\log n$ and stores it. Note that this
 requires only $O(\log k+\log \log n)\le O(\log\log n)$ nondeterministic bits and space.
Then $\B$ starts simulating $\A$. When $\A$ jumps, $\B$ guesses $\lceil\log (n+1)\rceil$ many bits encoding a number $m\in[n]$. It computes $f:=h_{p,q}(m)$ and checks that $f\notin X$.  Then it adds
$f$ to $X$, moves the input head to the $m$th input bit, changes to the starting state and continues 
the simulation of $\A$. 

Obviously, if $\A$ jumps at most $\ell$ times, then $\B$ uses at most $O(\log\log n + \ell\log n)$ nondeterministic bits.
To see that $\mathbb B$ runs in allowed space, observe that
the ``fingerprint'' $f$ can be computed in space $O(\log n)$: first
$b:=qm\text{ mod } p$ can trivially be computed in space polynomial in
$\log p$ and this is space $(\log\log n)^{O(1)}\le O(\log n)$; second,
$f=b\mod k^2$ can trivially be computed in space polynomial in $(\log
k+\log b)$ and 
%this is space 
the space usage her is $(\log\log n)^{O(1)}$.

We show that $\mathbb B$ accepts $x$ if and only if $x\in Q$. 
If $\B$ accepts $x$ then either because $\mathbb Q$ accepts $x$ 
(and then trivially $x\in Q$) or because $\A$ 
reaches an accepting state when it jumps to cells numbered $m_1,\ldots,m_\ell$; note that the fingerprints of 
these cell numbers are pairwise different, and hence so are the numbers.
This implies $x\in Q$. 
Conversely, if $x\in Q$, then there is an accepting run of $\A$ on $x$ with $\ell\le k$ jumps to pairwise different cells $m_1,\ldots,m_\ell$. By Lemma~\ref{lem:hash} there exist $q<p<k^2\log n$ such that $h_{p,q}$ is injective on $\{m_1,\ldots,m_\ell\}$. Then $\B$ accepts when first guessing some such pair $(p,q)$ and then strings encoding $m_1,\ldots, m_\ell$.
\end{proof}

%It is straightforward to program a Turing machine with jumps that solves %$\HOMP{\P^*}$: 
%on an instance $(\mathbf P^*_k,\relb)$ it jumps $k$ times to guess successively the %homomorphic images of the vertices of $\mathbf P^*_k$ in $B$; 
%to check edge relations and colours, the machine only needs to
%remember
%the most recent two guesses. Observe that
%the same program on a machine with injective jumps solves 
%$\EMP{\P^*}$. More generally, one can argue along these lines to show:

\begin{theorem}\label{theo:pathemb} Let $\fancya$ be a decidable class of structures of boun\-ded arity and of bounded pathwidth. Then $\EMP{\fancya}\in\PATH$.
\end{theorem}

\begin{proof} Choose a constant $w\in\mathbb N$ bounding the pathwidth of~$\fancya$. We use a machine $\A$ with injective jumps to solve $\EMP{\fancya}$. The result will then follow from Lemma~\ref{lem:jumps}.
 
Given an instance $(\rela,\relb)$ of $\EMP{\fancya}$ the machine first computes a width $\le w$ path-decomposition 
$(\mathbf P_k, (X_i)_{i\in[k]})$ of $\rela$ such that $X_{i}\subsetneq X_{i+1}$ or $X_{i+1}\subsetneq X_{i}$ for all $i\in[k-1]$; we further assume that no $X_i$ is empty. This is done in space effectively bounded in the parameter~$|\rela|$ 
%(in fact, space $|\rela|^{O(1)}$ is sufficient) 
and, in particular, $k$ is effectively bounded in~$|\rela|$.

It then computes inductively for each $i\in[k]$ a map $h_i $ from $X_i$ into $B$ 
that is a partial homomorphism from $\rela$ into~$\relb$. 
To start, the machine $\A$ jumps  $|X_1|$ times to guess elements $b_1,\ldots b_{|X_1|}\in B$. 
It checks that the function $h_1:X_1\to B$ that maps the $i$th element of $X_1$ to $b_i$ defines a partial homomorphism from $\rela$ into~$\relb$. 
Having computed $h_{i}$ the machine computes $h_{i+1}$ as follows. If $X_{i+1}\subsetneq X_i$, then $h_{i+1}:=h_i\upharpoonright X_{i+1}$ is the restriction of $h_i$ to~$X_{i+1}$. Otherwise $X_{i+1}\supsetneq X_i$, say $X_{i+1}=X_i\cup\{a_1,\ldots, a_d\}$; then $\A$ jumps $d$ times to guess $b_1,\ldots b_d\in B$ and checks that $h_{i+1}:=(h_i\upharpoonright X_i)\cup\{(a_j,b_j)\mid j\in[d]\}$  is a partial homomorphism from $\rela$ into $\relb$. In the end, if no check fails, $\A$ halts accepting.

This procedure can be implemented in pl-space: 
%Checking that $h_i$ is a partial homomorphism can be done by computing the canonical query of $\langle X_i\rangle^{\rela}$  (in space bounded in the parameter) and appealing to Lemma~\ref{lem:mcfo} to check that it is satisfied in $\relb$ by the assignment given by $h_i$; note the canonical query has width $|X_i|\le w$. 
the space to store the path decomposition is bounded in the parameter, and storing one $h_i$ needs space roughly $w\cdot (\log |A|+\log|B|)$. 

It is routine to check that $\A$ makes exactly $|A|$ many jumps, and that it accepts only if $\bigcup_i h_i$ is a homomorphism from $\rela$ to $\relb$. Since the machine has injective jumps it accepts in fact only if this homomorphism is an embedding. 
Conversely, it is obvious that the machine accepts if an embedding from $\rela$ into $\relb$ exists.
\end{proof}

\begin{proof}[Proof of Theorem~\ref{theo:pathcomplete}]
To see $\HOMP{\P^*}\in\PATH$, just consider the machine $\mathbb A$ described in the proof of Theorem~\ref{theo:pathemb} as a machine with jumps instead of as a machine with injective jumps.
 
To see that $\HOMP{\P^*}$ is hard for $\PATH$ under pl-reductions, 
let $(Q,\kappa)\in\PATH$ and choose a Turing machine $\A$ with jumps according Lemma~\ref{lem:jumps}~(2)  
that accepts~$Q$. We can assume that there are computable $f,g:\mathbb N\to\mathbb N$ such that
$\A$  on $x\in\{0,1\}^*$ runs in space
$O(g(\kappa(x))+\log |x|)$ and makes  on every run exactly $f(\kappa(x))$ many jumps.

Fix $x\in\{0,1\}^*$ and set $k:=\kappa(x)$ and $n:=|x|$.
Let $\A_\textit{det}$ be the deterministic Turing machine defined as $\A$ but with the jump state interpreted as a rejecting halting state. Observe that 
$\A_\textit{det}$ (and $\A$) has at most $m:=2^{g(k)}\cdot n^{c}$ configurations where 
 $c\in\mathbb N$ is a suitable constant. 
Let $c_1,\ldots,c_m$ be a list (possibly with repetitions) of all configurations of $\A_\textit{det}$ on $x$ 
whose state is the starting state. Assume that $c_1$ is the starting configuration of $\A_\textit{det}$.
For $i,j\in[m]$, say $i$ {\em reaches} $j$ if the computation of $\A_\textit{det}$ started on $c_i$ (with $x$ on the input tape) reaches in at most $m$ steps a configuration $c$ with the jump state,
%\hnote{slightly confusing, since in the det machine the jump state is being %``interpreted'' as a rejecting state}
%\mnote{it is correct as written, or? The $\A_\textit{det}$ machine is only %auxiliary.}
and $c_j$ is obtained from $c$ by changing the jump state to the starting state and changing  
the position of the input head to some arbitrary cell storing an input bit. 
Further, call $i\in[m]$ {\em accepting} if $\A_\textit{det}$ started on $c_i$ accepts  within at most $m$ steps.

Consider the structure $\relb_x$ given by
\begin{eqnarray*}
B_x&:=&[f(k)+1]\times[m],\\
E^{\relb_x}&:=& \textup{the symmetric closure of }\\
&& \{((i,j),(i+1,j'))\mid i\in[f(k)], j
\text{ reaches } j'  \},
\\
%\end{eqnarray*}
%\begin{eqnarray*}
C_1^{\relb_x}&:=&\{(1,1)\},\\
C_i^{\relb_x}&:=&\{i\}\times[m]\text{ for }2\le i\le f(k),\\
C_{f(k)+1}^{\relb_x}&:=&\{ (f(k)+1,j)\mid j\text{ is accepting}\}.
\end{eqnarray*} 
%\hnote{why do we take the symmetric closure here?}
%\mnote{because $\relp_k$ is undirected}
It is clear that there exists a homomorphism from $\relp_{f(k)+1}^*$ to
$\relb_x$ if and only if $\A$ accepts $x$, that is, the map $x\mapsto (\relp^{*}_{f(\kappa(x))+1},\relb_x)$ is a 
reduction from $(Q,\kappa)$ to $\HOMP{\P^*}$.
The new parameter $|\relp_{f(\kappa(x))+1}^*|$ depends only on $\kappa(x)$. The reduction is implicitly pl-computable: first observe that the numbers $f(k)$ and $m$ can be computed from $x$ in pl-space. A counter for numbers up to $m$ needs only space $O(g(k)+\log n)$. Hence one can tell whether or not $i$ reaches $j$ in pl-space simply by simulating $\A_\textit{det}$ for at most $m$ 
many steps. Similarly, this space is sufficient to tell whether or not a given $j\in[m]$ is accepting.
\end{proof}

The following result gives information about fundamental  problems: the problems
$\EMP{\dP}$, $\EMP{\C}$,
and $\EMP{\dC}$ 
 are the parameterized problems 
of determining if an input graph contains 
a simple directed $k$-path,
a simple undirected $k$-cycle,
and
 a simple directed $k$-cycle, 
respectively; 
these problems are denoted respectively by
{\sc $p$-DirPath}, 
{\sc $p$-Cycle}, and
{\sc $p$-DirCycle}
by Flum and Grohe~\cite{flumgrohecounting}.

\begin{theorem}\label{theo:pathproblems}
The following parameterized problems are complete for $\PATH$ under pl-reduc\-tions:\
\begin{quote}
$
\begin{array}{ll}
\pstpath,\\
\HOMP{\dP},& \EMP{\dP}\\ 
\HOMP{\C}, & \EMP{\C}\\
\HOMP{\dC}, & \EMP{\dC}
\end{array}
$
 
\end{quote}

\end{theorem}

\begin{proof}
By Theorem~\ref{theo:pathemb} all embedding problems are contained in $\PATH$. For the homomorphism problems 
and $\pstpath$
%and {\sc $p$-$st$-Path} 
the same argument works (see the proof of Theorem~\ref{theo:pathcomplete}). 
We are thus left to prove hardness.

Recall Example~\ref{exa:core}. Corollary~\ref{cor:coreemb} implies
that 
$\HOMP{\dP^*}\le_\pl\EMP{\dP}$ and also that $\HOMP{\dC^*}\le_\pl\EMP{\dC}.$
Since we trivially have $\HOMP{\fancya}\le_\pl\HOMP{\fancya^*}$ for all classes $\fancya$, we conclude that
$\HOMP{\dP}\le_\pl\EMP{\dP}$ and also that $\HOMP{\dC}\le_\pl\EMP{\dC}$.
For $\C$ we similarly get  
$\HOMP{\C_{\textup{odd}}}\le_\pl\EMP{\C_{\textup{odd}}}$ where $\C_{\textup{odd}}$ is the class of odd length cycles. By 
$\dC_{\textup{odd}}$ we denote the class of odd length directed cycles.

It thus suffices to show that the problems $$\HOMP{\dP}, \HOMP{\dC}, \HOMP{\C_{\textup{odd}}}, \pstpath$$
are $\PATH$-hard.
By Theorem~\ref{theo:pathcomplete}, we know that
$\HOMP{\P^*}$ is hard for $\PATH$.  We give the sequence of reductions
$$
\HOMP{\P^*} \le_\pl \HOMP{\dP}\le_\pl  
\pstpath 
\le_\pl \HOMP{\dC_{\textup{odd}}}
$$
and then show the hardness of 
$\HOMP{\C_{\textup{odd}}}$.\medskip

$\HOMP{\P^*} \le_\pl \HOMP{\dP}$. %\mnote{Corrected proof. Good?}
Let $(\relp_k^*,\relb)$ be an instance of $\HOMP{\P^*}$. 
The reduction produces the instance  $(\overrightarrow{\relp_k},\relb')$ where $\relb'$ is 
the directed graph with vertices $ B':=[k]\times B$ and edges
$$ E^{\relb'}:=\{((i,b),(i+1,b'))\mid i\in[k-1], b\in C^{\relb}_i, b'\in  C^{\relb}_{i+1}\}.
$$

$\HOMP{\dP}\le_\pl \pstpath$.
Let $(\overrightarrow{\relp_k},\relg)$ be an instance of
$\HOMP{\dP}$. The reduction produces the instance $(\relg', s, t, k+2)$ 
where $\relg'$ has vertices
$G':=\{s,t\}\cup ([k]\times G)$ and as edges the symmetric closure of
\begin{eqnarray*}
&&\big\{((i,u),(i+1,v))\mid i\in[k-1], (u,v)\in E^{\relg} \big\}\\
&&\cup\ \big(\{s\}\times([1]\times G)\big)\cup \big(\{t\}\times([k]\times G)\big).
\end{eqnarray*}

$\pstpath
\le_\pl \HOMP{\dC_{\textup{odd}}}$. Let $(\relg, s, t, k)$ be an instance of 
the former problem; by the previous reduction, we may assume that 
it is a yes instance if and only
if there is an $s$-$t$ path of length exactly $k$.
We can assume that $k$ is odd (otherwise we take a new neighbor of the given $s$ as our new $s$).
 Define the graph $\relg'$ with vertices $([k]\times G)$ and edges 
as follows.  When $i \in [k-1]$ and $(u, v) \in E^{G}$, there is an edge
from $(i, u)$ to $(i+1, v)$; also, there is an edge from 
$(k, t)$ to $(1, s)$.
Then $(\relg,s, t, k)  \mapsto  (\overrightarrow{\mathbf C_k}, \relg')$ is a reduction as desired.\medskip

%from $(i,u)$ to $((i+1)\text{mod}\ k ,v)$ if $(u,v)\in E^G$ and one of the %following holds:
% \begin{enumerate}\itemsep=0pt
%\item[--] $2\le i\le k-1$, or
%\item[--] $i=1$ and $u=s$, or
%\item[--] $i=k$ and $u=t$ and $v=s$.
%\end{enumerate}

Finally, we show the hardness of 
$\HOMP{\C_{\textup{odd}}}$. 
By appeal to Lemma~\ref{lemma:reduction-constants-cores}, 
it suffices to demonstrate a reduction
$\pstpath
\le_\pl \HOMP{\C^*_{\textup{odd}}}$. 
Given an instance $(\relg, s, t, k)$ of the former problem 
of the
above form, 
we define
$\relg'$ as in the previous reduction.  The produced instance is
$({\mathbf C_k^*}, \relg'')$, where $\relg''$ is the expansion of 
the symmetric closure of $\relg'$
with $C_i^{\relg''} = \{ i \} \times G$.
\end{proof}

\section{The class TREE}
\label{sect:tree}

We give a machine characterization of the class of parameterized problems that are pl-reducible to $\HOMP{\T^*}$. 
%In this extended abstract, we merely state two central theorems concerning this class.

\begin{definition}\label{def:treemachine} 
The class $\TREE$ contains 
a parameterized problem $(Q,\kappa)$ if there are
a computable function $f:\mathbb N\to\mathbb N$ 
and 
an
alternating Turing machine that accepts $Q$,
is pl-space bounded with respect to $\kappa$,
and uses $f(\kappa(x))\cdot\log |x|$ nondeterministic bits and $f(\kappa(x))$ co-nondeterministic bits.
\end{definition}

%It is easy to see that $\TREE$ is closed under pl-reductions.
The following proposition is straightforward to verify.

\begin{prop} 
The complexity class $\TREE$ is closed under pl-reductions.
\end{prop}

\begin{definition}\label{def:acc}
An {\em alternating Turing machine with jumps} is a Turing machine $\A$ using nondeterministic jumps and a universal guess state (see Preliminaries). It {\em accepts} an input $x\in\{0,1\}^*$ if its starting configuration on $x$ is {\em accepting}: 
it is already explained what an accepting halting configuration is, and a non-halting configuration which is not in the universal guess state (resp. is in the universal guess state) is accepting if at least one (resp. both) of its successor configurations are accepting. 
\end{definition}

\begin{lemma}\label{lem:jumps2} Let $(Q,\kappa)$ be a parameterized problem. The following are equivalent.
\begin{enumerate}\itemsep=0pt
\item $(Q,\kappa)\in\TREE$ 
\item There exists a computable $f:\mathbb N\to\mathbb N$ and an alternating  Turing machine $\A$ with $f\circ \kappa$ many jumps and $f\circ \kappa$ many co-nondeterministic bits that accepts $Q$ and is pl-space bounded with respect to $\kappa$.
\end{enumerate}
\end{lemma}

\begin{proof} The implication from (1) to (2) can be seen analoguously to the corresponding implication in Lemma~\ref{lem:jumps}. 

Conversely, let $\A$ and $f$ accord (2). A machine $\B$ according (1) can be obtained by simulating a jump of $\A$ by existentially guessing a binary string encoding a number $m\in[n]$ and moving the input head to cell~$m$. 
\end{proof}

\begin{theorem}\label{theo:treecomplete} 
$\HOMP{\T^*}$ is complete for $\TREE$ under pl-reductions.
\end{theorem}

%For $k\ge 0,$ the structure $\overrightarrow{\mathbf B_k}$ has
%universe  $\{0, 1 \}^{\leq k}$ and binary 
%relations $S_i^{\overrightarrow{\mathbf B_k}} = \{ (x, xi)\mid x \in \{ 0, 1 \}^{\le %k-1} \}$ for $i\in\{0,1\}$.
%%The class $\dBIN$  consists of the structures that are isomorphic to a structure of this form.
%Let $\mathbf B_k$ be the graph underlying the directed graph 
%with universe $\{ 0, 1 \}^{\leq k}$ and edges $S_0^{\overrightarrow{\mathbf B_k}}\cup S_1^{\overrightarrow{\mathbf B_k}})$.
%%The class $\BIN$ of {\em complete binary trees} 
%%consists of the structures that are isomorphic to a structure of this form.
%%
%Here, we write $\{0,1\}^{\le k}$ for the set of binary strings $x\in\{0,1\}^*$ of length $|x|\le k$.\medskip

\begin{proof}(Theorem~\ref{theo:treecomplete}) We show that $\HOMP{\T^*}\in\TREE$. Consider the following alternating Turing machine. Given an instance $(\relt,\relb)$ of $\HOMP{\T^*}$, the machine chooses some $t\in T$ as a ``root'' and computes the directed ``tree'' $\relt'$ with edges directed away from $t$. It existentially guesses ($O(\log |B|)$ bits encoding) a $b\in C_t^B$ and writes $(t,b)$ on some tape. While the pair $(t,b)$ written on the tape is such that $t$ has children in $\relt'$ the machine does the following: universally guess ($O(\log |T|)$ bits encoding) a child $t'$ of $t$; 
existentially guess $b'\in B$; check that $(b,b')\in E^B$ and $b'\in C^B_{t'}$. The while loop is left rejecting if this check fails. If the machine leaves the while loop otherwise, it accepts.

The number of universal guesses is bounded by $O(|T|\cdot\log|T|)$. The number of existential guesses is 
bounded by $|T|\cdot\log|B|$. The machine uses space to store $\relt'$ and at most two pairs in $T\times B$, so it is pl-space bounded.

To show $\HOMP{\T^*}$ is TREE-hard under pl-reductions, let the parameterized problem $(Q,\kappa)$ be in $\TREE$. Choose an alternating  machine $\A$ with jumps according to Lemma~\ref{lem:jumps2} for $(Q,\kappa)$. By adding some dummy jumps and dummy universal guesses we can assume that $\A$ on every $x$ and every run on $x$
 first makes one universal guess, then one jump, then one universal guess and so on. We can further assume that $\A$ on $x$ on every run on $x$ makes exactly $f(\kappa(x))$ many jumps and exactly $f(\kappa(x))$ many universal guesses.
Let $\A^0$ ($\A^1$) be the machine obtained from $\A$ by fixing the transition from a configuration with universal guess state to the first (second) successor configuration. Note $\A^0$ and $\A^1$ are Turing machines with jumps.

Let $x\in\{0,1\}^*, k:=\kappa(x) $, $n:=|x|$. Recall the proof of Theorem~\ref{theo:pathcomplete}. As there, let $c_1,\ldots,c_m$ enumerate all configurations of $\A$ on $x$ with the starting state; assume $c_1$ is the starting configuration.
Let  $\A^0_{\textit{det}}$ and $\A^1_{\textit{det}}$ be the deterministic machines obtained form $\A^0$ and 
$\A^1$ by interpreting the jump state as a rejecting halting state. For $i,j\in[m],b\in\{0,1\}$ we 
define what it means that $i$ {\em $b$-reaches} $j$ as in the proof of Theorem~\ref{theo:pathcomplete} with $\A^b_\textit{det}$ in place of $\A_\textit{det}$ there; call $i$ {\em accepting} if $\mathbb A_\textit{det}^0$ (equivalently  $\mathbb A_\textit{det}^1$) started on $c_i$ accepts in at most $m$ steps without entering the universal guess state.

Recall the notation $\mathbf T_{k}$ from the preliminaries.
 The reduction outputs $(\relt^*_{f(k)+1},\relb)$ where
$\relb$ is defined as follows.
\begin{align*}
B:=\ &\{0,1\}^{\le f(k)+1}\times[m],\\
E^{\relb}:=\  & \textup{ the symmetric closure of } \\
\ &\big\{((\sigma,j),(\sigma b,j'))\mid b\in\{0,1\},\sigma\in\{0,1\}^{\le f(k)}, j \text{ $b$-reaches } j' \big \},\\
C_{\lambda}^{\relb}:=\ &\{(\lambda,1)\},\text{ where $\lambda$ is the empty string},\\
C_{\sigma}^{\relb}:=\ &\{\sigma\}\times[m], \text{ for }1\le |\sigma|\le f(k),\\
C_{\sigma}^{\relb}:=\ &\{ (\sigma,j)\mid j \text{ is accepting}\}, \text{ for }|\sigma|=f(k)+1.
\end{align*}
It is not hard to see that $(\relt^*_{f(k)+1},\relb)$ can be computed in pl-space (cf. Proof of Theorem~\ref{theo:pathcomplete}). To see this indeed defines a reduction, first assume
$h$ is a homomorphism from $\relt^*_{f(k)+1}$ to $\relb$. As $h$ preserves the unary relations $C_{\sigma}$, for every $\sigma$ there is an $i_\sigma\in[m]$ such that $h(\sigma)=(\sigma,i_\sigma)$. It follows by induction on $\ell$ that for every $\sigma\in\{0,1\}^{f(k)+1-\ell}$ the configuration $c_{i_{\sigma}}$ is accepting (Definition~\ref{def:acc}).  But $i_\lambda=1$, so $c_{i_{\lambda}}=c_1$ is the starting configuration and $\A$ accepts~$x$.

Conversely, assume  $\A$ accepts $x$. We define an accepting configuration $c_\sigma$ for every $\sigma\in \{0,1\}^{\le f(k)+1}$: $c_\lambda$ is the starting configuration $c_1$. All other $c_\sigma$s are going to be the result of a jump (are a successor of a configuration in the jump state). Assume $c_\sigma$ is already defined. Then $c_\sigma$ is the starting configuration or results from a jump. In both cases the machine  $\A$ reaches from  $c_\sigma$  deterministically a universal guess state with two accepting successors $c'_0,c'_1$. For every $b\in\{0,1\}$, $\A$ reaches deterministically from $c'_b$ either an accepting halting configuration or a configuration in the jump state. In the first case let $c_{\sigma b}$ be this accepting halting configuration and in the second let it be some accepting successor of the jump. 
For every $\sigma$ choose $i_{\sigma}\in[m]$ such that $c_\sigma=c_{i_\sigma}$. Then $\sigma\mapsto (\sigma,i_\sigma)$ defines a homomorphism from $\relt^*_{f(k)+1}$ to~$\relb$.\end{proof}

\begin{theorem} 
\label{theorem:emb-in-tree}
Let $\fancya$ be a decidable class of structures of boun\-ded arity and bounded treewidth. Then $\EMP{\fancya}\in\TREE$.
\end{theorem}

%NOTE: following problems should also be complete for
%$\TREE$: 
%$\HOMP{\BIN^*}$, $\EMP{\BIN^*}$
%$\HOMP{\dBIN}$, $\EMP{\dBIN}$

\begin{proof} Let $\fancya$ accord the assumption. We proceed as in the proof of Theorem~\ref{theo:embtd}.\medskip

\noindent{\em Claim.} There exists a decidable class of connected structures $\fancya'$ of
 bounded treewidth such that $\EMP{\fancya}\le_\pl\EMP{\fancya'}$.\medskip

Note $\EMP{\fancya'}\le_\pl\HOMP{(\fancya')^*}$ by Lemma~\ref{lem:connected}, the latter problem pl-reduces to
$\HOMP{\mathcal T^*}$ by the Classification Theorem, and $\HOMP{\mathcal T^*}\in\textup{TREE}$ by Theorem~\ref{theo:treecomplete}.
We are thus left to prove the claim.

 Assume $\fancya$ has treewidth at most $w$.
Fix a computable function that maps every $\rela\in\fancya$ to a width $\le w+1$ tree decomposition 
$(\relt, (X_t)_{t\in T})$ of $\rela$ such that $|X_t|\ge 2$ for all $t\in T$, and $X_s\cap X_t\neq\emptyset$ for all $(s,t)\in E^{\relt}$. Let
$\rela'$ be the expansion of $\rela$ by interpreting a new binary relation symbol $R$ by $\bigcup_{t\in T}X_t^2$. 
Then $(\relt, (X_t)_{t\in T})$ is also a tree decomposition of $\rela'$ and $\rela'$ is connected. Clearly,
$\fancya':=\{\rela'\mid\rela\in\fancya\}$ is decidable. The map $(\rela,\relb)\mapsto(\rela',\relb')$, where $\relb'$ 
is the expansion of $\relb$ interpreting $R$ by $B^2$, is a pl-reduction from $\EMP{\fancya}$ to $\EMP{\fancya'}$.
\end{proof}

\begin{theorem}
\label{thm:hom-hard-for-tree}
The parameterized problems $\HOMP{\BIN}$,
$\HOMP{\dBIN}$,
$\EMP{\BIN}$, and
$\EMP{\dBIN}$
are complete for $\TREE$ under pl-reduc\-tions.
\end{theorem}

\begin{proof} 
It is straightforward to verify that the structures in $\BIN$ 
and in $\dBIN$ are connected cores.
Hence, each of the first two problems is $\TREE$-complete by 
the Classification Theorem
and
Theorem~\ref{theo:treecomplete}.

The problems $\EMP{\BIN}$ and $\EMP{\dBIN}$ are $\TREE$-hard by
Corollary~\ref{cor:embhom} and the hardness of 
$\HOMP{\BIN}$ and $\HOMP{\dBIN}$, which immediately imply the hardness of
$\HOMP{\BIN^*}$ and $\HOMP{\dBIN^*}$.

The problems $\EMP{\BIN}$ and $\EMP{\dBIN}$ are in $\TREE$
by Theorem~\ref{theorem:emb-in-tree}.
\end{proof}

\section{Counting classification}
\label{sect:counting}

\newcommand{\fancyd}{\mathcal{D}}
\newcommand{\fancyt}{\mathcal{T}}

In this section we present a classification of the counting problems corresponding to
 the problems $\HOMP{\fancya}$.

\subsection{Preliminaries on parameterized counting complexity}
A machine with oracle $O\subseteq\{0,1\}^*$ 
has an extra write-only {\em oracle tape};
such a machine has a
{\em query state} and the word $y$ written on the oracle tape is the {\em query} of a configuration with this state; 
the successor state is obtained by erasing the oracle tape and moving to one of two distinguished states depending of 
whether the query is contained in the oracle $O$ or not. The oracle tape is not accounted for in space bounds (as in~\cite{ladnerlynch}).

A {\em parameterized counting problem} is a pair $(F,\kappa)$ of a function $F:\{0,1\}^*\to\mathbb N$ and a parameterization $\kappa$. 
To say it is in para-L, means that $F$ is implicitly pl-computable with respect to $\kappa$. Equivalently one could say that there is a
Turing machine with a write-only output tape that computes $F$ and is pl-space bounded with respect to~$\kappa$.  

A {\em parsimonius fpt-reduction} from $(F,\kappa)$ to another parameterized counting problem $(F',\kappa')$ is a function $R:\{0,1\}^*\to\{0,1\}^*$ that is computable by an fpt-time bounded (with respect to $\kappa$) Turing machine such 
that $F=F'\circ R$ and $\kappa'\circ R\le f\circ \kappa$ for some computable $f:\mathbb N\to\mathbb N$. 
In the logspace setting we define a {\em parsimonious pl-reduction} similarly demanding that the reduction is 
implicitly pl-computable instead of computable by a fpt-time bounded machine. We again write $(F,\kappa)\le_\pl(F',\kappa')$ if  such a reduction exists. 

We say $(F,\kappa)$ is {\em pl-Turing reducible} to $(F',\kappa')$ and write $(F,\kappa)\le^T_\pl (F',\kappa')$ if there are 
a pl-space bounded (with respect to $\kappa$) Turing machine $\mathbb A$ 
with oracle to $\textsc{Bitgraph}(F')$ that decides $\textsc{Bitgraph}(F)$, and a computable $f$ such that on every input $x\in\{0,1\}^*$ all 
queries $y\stackrel{?}{\in} \textsc{Bitgraph}(F')$ 
of $\A$ on $x$ have parameter $\kappa'(y)\le f(\kappa(x))$. Here, we denote the parameterizations of $\textsc{Bitgraph}(F)$ and $\textsc{Bitgraph}(F')$ again by $\kappa$ and $\kappa'$ respectively.

%We write $(Q,\kappa)\equiv^T_\pl (Q',\kappa')$ if both
%$(Q,\kappa)\le^T_\pl (Q',\kappa')$ and $(Q',\kappa')\le^T_\pl (Q,\kappa)$.

%For a class of structures $\fancya$, the problem $\HOMP{\fancya}$  has the following associated  parameterized counting problem
%
%\npprob{$\SHOMP{\fancya}$}{a pair of structures $(\rela,\relb)$ with $\rela\in\fancya$}{$|\rela|$}{compute the number of homomorphisms from $\rela$ to $\relb$.}
%
%More formally, this is the parameterized counting problem $(F,\kappa)$ where $F$ maps a pair of structures $(\rela,\relb)$ with $\rela\in\fancya$ to the number of homomorphisms from $\rela$ to $\relb$, and maps an argument not of this form to, say, 0; the parameterization $\kappa$ maps any pair of structures to the size of the first component and other arguments to 0.

\subsection{Classification theorem}

For a 
class of structures $\fancya$ consider the parameterized counting problem.
\npprob{$\SHOMP{\fancya}$}{A pair of structures $(\rela, \relb)$
where  $\rela\in\fancya$}{$|\rela|$}{Compute the number of homomorphisms from $\rela$ to $\relb$.}
Dalmau and Jonsson~\cite{dalmaucounting} gave a classification
of counting problems of this form, showing that
for a class of structures $\fancya$ of boun\-ded arity,
the problem $\SHOMP{\fancya}$ is in FPT if $\fancya$ has bounded
tree\-width,
and is $\textsc{\#}\textup{W[1]}$-complete otherwise.
We give a fine classification of the case where $\fancya$ has
bounded treewidth, analogous to our fine classification for the
problem 
$\HOMP{\fancya}$.

\begin{comment}
We let it be understood that such a problem is in para-L if and only if the
desired number (in binary) is implicitly pl-computable.
We compare the defined counting problems 
via parsimonious pl-reductions $\le_\pl$ and via 
 parameterized logarithmic space Turing reductions $\le^T_\pl$.
These notions are straightforward logspace analogues of the common
notions of parsimonious fpt-reductions and fpt-Turing
reductions~\cite{flumgrohe}. 
\end{comment}
% We give definitions in the following subsection.

%We have the following classification.

\begin{theorem}[Counting Classification]
\label{thm:counting-classification}
Let $\fancya$ be a decidable class of structures 
having bounded arity and bounded treewidth.
\begin{enumerate}
\item If $\fancya$ has unbounded pathwidth,
then 
$$\SHOMP{\fancya} \le_\pl \SHOMP{\T^*} \le^T_\pl \SHOMP{\fancya}.$$

\item If $\fancya$ has bounded pathwidth and unbounded tree depth,
then 
$$\SHOMP{\fancya} \le_\pl \SHOMP{\P^*} \le^T_\pl \SHOMP{\fancya}.$$

\item If $\fancya$ has bounded tree depth, then 
$$\SHOMP{\fancya}\in\textup{para-L}.$$
\end{enumerate}
\end{theorem}

The proof of this result partly involves an analysis of the proof of
Theorem~\ref{thm:classification}, 
and builds on techniques of Dalmau and 
Jonsson~\cite{dalmaucounting}.

\begin{lemma}
\label{lemma:counting-reductions}
Let $\fancya$ be a decidable set of finite structures, 
let $\fancyg$ be the set of Gaifman
graphs of $\fancya$, and let $\fancym$ be the set of minors of graphs
in $\fancyg$. Then
\begin{eqnarray*}
&&\SHOMP{\fancym^*} \le_{\pl} \SHOMP{\fancyg^*} 
\le_{\pl} \SHOMP{\fancya^*}
\le^T_{\pl} \SHOMP{\fancya}.
\end{eqnarray*}
%We view each of these problems, 
%other than the last one, as a promise problem.
\end{lemma}

\begin{proof}
The first two reductions are exactly as before, that is,
they are the reductions from Lemmas~\ref{lemma:red-minors-to-graphs}
and~\ref{lemma:red-graphs-to-structures}. 
These reductions are readily verified
to be parsimonious. We thus prove that $\SHOMP{\fancya^*} \le_{\pl}^T \SHOMP{\fancya}$.

Let $\rela$ be an element of $\fancya$, and
let $(\rela^*, \relb)$ be an instance of
$\SHOMP{\fancya^*}$.
Let $\relb_0$ be the restriction of $\relb$ to 
relation symbols of $\rela$.
For each non-empty subset $S \subseteq A$,
define $\relb_S$ to be the induced substructure
of $\rela \times \relb_0$
on universe $\{ (a, b) \in S \times B\mid b \in C_{a}^{\relb} \}$.
For a mapping $g$ from $A$ to a set of the form $B_S$,
let $g_1$ denote the map $(\pi_1 \circ g)$ where $\pi_1$ is the projection of a pair to its first component.

The number of homomorphisms $g$ from $\rela^*$ to $\relb$ is the same as the number $M_g$ of homomorphisms $g'$ from $\rela$ to $\relb_A$ such that
$g'_1$ is the identity on $A$ (consider the bijection $g\mapsto g'$
with $g'(a):=(a,g(a))$).
%\mnote{corrected, old outcommented}
%
%It is straightforward to verify that
%a mapping $g: A \rightarrow B$ is a homomorphism from 
%$\rela^* $ to $\relb$ 
%if and only if it
%is a homomorphism 
%$\rela \rightarrow \relb_A$ with $g_1$ equal to the identity.
%
%
To compute $M_g$, it suffices to compute the number $M_h$ of homomorphisms 
$h$ from $ \rela$ to $\relb_A$ such that
$h_1(A) = A$.
This is because of the fact that a mapping $h: A \rightarrow A \times B$
is a homomorphism from $\rela$ to $\relb_A$
with
$h_1(A) = A$
if and only if 
$h$ has the form $g \circ \sigma$ where 
$g$ is a homomorphism from $\rela $ to $\relb_A$,
$g_1$ is the identity,
and $\sigma$ is a bijective homomorphism from $\rela $ to~$\rela$.
From this fact, it follows that $M_g = M_h / S$ where $S$
is the number of bijective homomorphisms from $\rela$ to~$\rela$;
note that $S$ can be computed directly from $\rela$, and so this gives
a way to determine $M_g$ (division is logspace computable~\cite{chiu}).
We prove the claimed fact as follows.
The backward direction is clear, so we prove the forward direction.
Let $h$ be a homomorphism from $\rela $ to $\relb_A$
with $(\pi_1 \circ h)(A) = A$.
There exists an integer $m \geq 1$ such that
$(\pi_1 \circ h)^m$ is the identity mapping on $A$.
Set $g = h \circ (\pi_1 \circ h)^{m-1}$; we then have 
$h = h \circ (\pi_1 \circ h)^m = g \circ (\pi_1 \circ h)$, as desired.

For each subset $S \subseteq A$,
the Turing reduction will query the instance $(\rela, \relb_S)$
of $\SHOMP{\fancya}$; denote the result by $N_{\subseteq S}$.
Observe that $N_{\subseteq S}$ is the number of homomorphisms
$h$ from $\rela$ to $\relb_A$ with $h_1(A) \subseteq S$.
For a subset $S \subseteq A$, let $N_{= S}$ denote the number
of homomorphisms 
$h$ from  $\rela$ to $\relb_A$ with $h_1(A) = S$.
We have, for each subset $S \subseteq A$, the identity
$N_{\subseteq S} = \sum_{T \subseteq S} N_{= T}$.
By inclusion-exclusion, we have
$N_{=A} = \sum_{S \subseteq A} (-1)^{|A| - |S|} N_{\subseteq S}$
which is the value $M_h$ that we wanted to determine.
We can evaluate the sum expression in pl-space
by combining two observations: first, 
with oracle access to $\SHOMP{\fancya}$ the sequence of the numbers
 $(-1)^{|A| - |S|} N_{\subseteq S}$ is implicitly pl-computable;
second, summing a sequence of integers can be done in logspace.
\end{proof}

\begin{proof}[Proof of Theorem~\ref{thm:counting-classification}]
Statements (1) and (2) each make two claims. The claims made first concern parsimonious pl-reductions and follow from Lemma~\ref{lemma:tree-decomp-to-tree} and Remark~\ref{rem:count}. The second claims concern Turing reductions and follow from the previous lemma together with the Excluded Tree Theorem~\ref{thm:rs}~(2) and the Excluded Path Theorem~\ref{thm:rs}~(3) respectively. 

We are left to prove Statement (3). It is not hard to see that a structure of tree depth at most $w'$ has a tree decomposition of width at most $w'+1$ such that the underlying tree has height at most $w'$ with respect to some root. By Lemma~\ref{lemma:tree-decomp-to-tree} and Remark~\ref{rem:count}, it suffices to show $\SHOMP{\mathcal T(w)}\in\textup{para-L}$ for every $w\in \mathbb N$. Here, we let $\mathcal T(w)$ be the class of structures $\relt^*$
such that $\relt$ is a tree that can be rooted in such a way that its height 
is at most $w$.

%By appeal to XXX and Lemma~\ref{lemma:tree-decomp-to-tree},
%we can assume that each structure in $\fancya$ is of the form
%$\relt^*$ where $\relt$ is a tree,
%and that there is a constant $k \geq 0$ such that for each
%$\relt^*$
%in $\fancya$, 
%the tree $\relt$ can be rooted in such a way that its depth is bounded
%above by $k$.
%
%For $w \geq 0$, let $\mathcal T(w)$ denote the set of structures of the form
%\relt^*$ where $\relt$ is a tree that is rootable in a way that
%its depth is less than or equal to $d$.
%We prove, by induction on $d$, that the problem
%$\SHOMP{\fancyd_d}$ is in para-L.  This suffices, since
%$\fancya \subseteq \fancyd_k$.

For $w = 0$, this is easy to see. So let $w > 0$ and assume by induction that 
$\SHOMP{\fancyt(w-1)}\in\textup{para-L}$.
Given an instance $(\relt^*, \relb)$ of $\SHOMP{\fancyt(w)}$,
we conceive of $\relt^*$ as a rooted tree with root $r$
and of height at most $w$.
For elements $t \in T$ and $b \in B$, we define
$N_{t \rightarrow b}$ to be the number of partial homomorphisms
$h$ that are defined on the subtree rooted at $t$ and such that $h(t) = b$.
Let $t_1, \ldots, t_m$ denote the children of $r$ in $\relt$.
The number that we desire to determine is
 $\sum_{b \in C_r^{\relb}} N_{r \rightarrow b}$.
For a particular value $b \in B$, 
it is straightforward to verify that 
$$\textstyle
N_{r \rightarrow b} = 
\prod_{i=1}^m \sum_{b'}   N_{t_i \rightarrow b'}$$
where the sum is over all $b' \in C_{t_i}^{\relb}$
such that $(b, b') \in E^{\relb}$. Thus, the number we desire to compute equals a certain sum-product-sum expression. But this expression is implicitly pl-computable: to determine bits of the numbers
$N_{t_i \rightarrow b'}$ one can run an algorithm witnessing $\SHOMP{\fancyt(w-1)}\in\textup{para-L}$.
Using the facts that iterated sum and iterated product are computable
in logarithmic space~\cite{chiu}, 
it follows that our sum-product-sum expression
can be evaluated in logarithmic space. This yields  the result.
\end{proof}

\section{Discussion}

A classification of the parameterized complexity of embedding problems is famously open \cite[p.355]{flumgrohe}, in particular, it is not known whether the embedding 
problem for complete bipartite cliques is W[1]-hard (under fpt-reductions). 

A fundamental problem whose complexity we failed to settle within our framework is $\EMP{\P}$. 
By Theorem~\ref{theo:pathemb}, we know $\EMP{\P}\in\PATH$, but we do not know whether it is $\PATH$-hard (under pl-reductions). 
We note that its restriction to regular graphs is in para-L.
\npprob{$\EMP{\P}_{\textup{reg}}$}{A regular graph $\relg$ and $k\in\mathbb N$}{$k$}{Does $\relg$ contain a path of length $k$?}

\begin{proposition}\label{theo:regpath}
 $\EMP{\P}_{\textup{reg}}\in \textup{para-L}$.
\end{proposition}
The proof uses a result of Flum and Grohe~\cite[Example~6]{describing} stating that model checking first-order logic on 
bounded degree graphs is in $\textup{para-L}$. 
Their proof actually shows 
\begin{theorem}[\cite{describing}] 
\label{thm:fo-on-bounded-degree}
The following parameterized problem is in $\textup{para-L}$.
\em
\npprob{}{A graph $\mathbf G$ of degree at most $d$ and a first-order sentence $\varphi$}{$d+|\varphi|$}{$\mathbf G\models\varphi$ ?}
\end{theorem}

\begin{proof}[Proof of Proposition~\ref{theo:regpath}]
Given a regular graph $\relg$ of degree $d$ and a natural $k\in\mathbb N$, distinguish two cases: if $d>k$ then accept; otherwise 
check, using the algorithm of
Theorem~\ref{thm:fo-on-bounded-degree},
whether $\relg$ satisfies $ \exists x_0\cdots x_k\big(\bigwedge_{i<j\le k}\neg x_i=x_j \wedge \bigwedge_{i< k}Ex_ix_{i+1} \big)$.
\end{proof}

On the more structural side, 
as mentioned in the introduction,
we believe that it could be worthwhile to
investigate whether or not the classes $\PATH$  and $\TREE$ are
closed under complement. 
Relatedly, one can ask whether or not it holds that
 $\textup{co-PATH}\subseteq\TREE$.

\subsection*{Acknowledgements} 
The first author was supported 
by the Spanish Project FORMALISM (TIN2007-66523), 
by the Basque Government Project S-PE12UN050(SAI12/219), and 
by the University of the Basque Country under grant UFI11/45.
The second author thanks the FWF (Austrian Science Fund) for its
support through Project P~24654~N25.

\bibliographystyle{abbrv}
\bibliography{refined,hubiebib}

\end{document}